\documentclass[conference]{IEEEtran}

\usepackage{booktabs}
\usepackage[numbers]{natbib}
\usepackage{enumerate}
\usepackage{amsmath,amsthm,xspace,amssymb}
\usepackage{latexsym}  %
\usepackage{wrapfig}
\usepackage{chngcntr}
\usepackage{ifthen}
\usepackage{hyperref}
\usepackage{multirow}
\usepackage[textwidth=2.1cm,textsize=small,disable]{todonotes}
\usepackage{scalerel}

\usepackage{soul}
\usepackage[d]{esvect}

\usepackage{xhfill}
\newcommand{\ditto}[1][10pt]{\rule[0.25em]{#1}{.4pt}~\texttt{"}~\rule[0.25em]{#1}{.4pt}}

\usepackage{tikz}

\usetikzlibrary{shapes,arrows,arrows.meta,calendar,matrix,calc,decorations,snakes}
\tikzset{wave/.style={decorate, decoration={snake, segment length=2mm, amplitude=0.3mm}}}

\hypersetup{linkcolor=blue!70!red,colorlinks=true}

\input xy
\xyoption{all}
\xyoption{2cell}
\UseAllTwocells

\newtheorem{observation}{Remark}[section]
\newtheorem{lemma}[observation]{Lemma}  %
\newtheorem{theorem}[observation]{Theorem}

\theoremstyle{definition}
\newtheorem{definition}[observation]{Definition}
\newtheorem{assumption}[observation]{Assumption}

\newtheorem{example}[observation]{Example}
\theoremstyle{remark}
\newtheorem{remark}[observation]{Remark}

\theoremstyle{plain}
\newtheorem{proposition}[observation]{Proposition}
\newtheorem{corollary}[observation]{Corollary}

\makeatletter
\protected\def\specialmergetwolists{%
  \begingroup
  \@ifstar{\def\cnta{1}\@specialmergetwolists}
    {\def\cnta{0}\@specialmergetwolists}%
}
\def\@specialmergetwolists#1#2#3#4{%
  \def\tempa##1##2{%
    \edef##2{%
      \ifnum\cnta=\@ne\else\expandafter\@firstoftwo\fi
      \unexpanded\expandafter{##1}%
    }%
  }%
  \tempa{#2}\tempb\tempa{#3}\tempa
  \def\cnta{0}\def#4{}%
  \foreach \x in \tempb{%
    \xdef\cnta{\the\numexpr\cnta+1}%
    \gdef\cntb{0}%
    \foreach \y in \tempa{%
      \xdef\cntb{\the\numexpr\cntb+1}%
      \ifnum\cntb=\cnta\relax
        \xdef#4{#4\ifx#4\empty\else,\fi\x#1\y}%
        \breakforeach
      \fi
    }%
  }%
  \endgroup
}
\makeatother

\tikzset{
    ncbar angle/.initial=90,
    ncbar/.style={
        to path=(\tikztostart)
        -- ($(\tikztostart)!#1!\pgfkeysvalueof{/tikz/ncbar angle}:(\tikztotarget)$)
        -- ($(\tikztotarget)!($(\tikztostart)!#1!\pgfkeysvalueof{/tikz/ncbar angle}:(\tikztotarget)$)!\pgfkeysvalueof{/tikz/ncbar angle}:(\tikztostart)$)
        -- (\tikztotarget)
    },
    ncbar/.default=0.5cm,
}

\tikzset{square left brace/.style={ncbar=0.15cm}}
\tikzset{square right brace/.style={ncbar=-0.15cm}}

\tikzstyle{w}=[draw,fill,circle,inner sep=1]
\tikzstyle{m}=[node distance=0.35cm]
\tikzstyle{p}=[node distance=0.15cm]

\def\perclr{blue}
\def\actclr{red}

\def\changed#1{#1} %

\newif\iffullver
\fullvertrue

\iffullver
\RequirePackage[T1]{fontenc}
\RequirePackage[tt=false, type1=true]{libertine}
\RequirePackage[varqu]{zi4}
\RequirePackage[libertine]{newtxmath}
\fi

\newcounter{runctr}
\newcounter{runctrplusone}
\def\run#1#2#3#4{
      \setcounter{runctr}{0}
     \whiledo{\numexpr\value{runctr}<#1}{
            \node[w] (w#2\therunctr) at ($ (#3,#4) + ( 0, \therunctr ) $) {};
            \node[m,left of=w#2\therunctr,anchor=east] (m#2\therunctr) {};
            \stepcounter{runctr}
        }
      \node (w#2\therunctr) at ($( #3, #4)+ (0,\therunctr)$) {};
      \setcounter{runctr}{0}
      \setcounter{runctrplusone}{1}
     \whiledo{\numexpr\value{runctr}<#1}{
            \draw[->] (w#2\therunctr)-- node (s#2\therunctr){} (w#2\therunctrplusone);
            \stepcounter{runctr}
            \stepcounter{runctrplusone}
        }
  }
\def\runfinite#1#2#3#4{
      \setcounter{runctr}{0}
     \whiledo{\numexpr\value{runctr}<#1}{
            \node[w] (w#2\therunctr) at ($ (#3,#4) + ( 0, \therunctr ) $) {};
            \node[m,left of=w#2\therunctr,anchor=east] (m#2\therunctr) {};
            \stepcounter{runctr}
        }
      \setcounter{runctr}{0}
      \setcounter{runctrplusone}{1}
     \whiledo{\numexpr\value{runctrplusone}<#1}{
            \draw[->] (w#2\therunctr)-- node (s#2\therunctr){} (w#2\therunctrplusone);
            \stepcounter{runctr}
            \stepcounter{runctrplusone}
        }
  }

\tikzset{>=Stealth}

\def\t{\mathrm{T}}

\def\KC{K^C}
\def\KP{K^P}
\def\K{K}
\def\WC{W^C}
\def\WP{W^P}

\def\B#1{[{#1}]}
\def\D#1{\langle{#1}\rangle}

\renewcommand\phi{\varphi}
\def\vec#1{\vv{#1}}

\begin{document}

\title{\huge Security Properties through the Lens of Modal Logic}

\author{\IEEEauthorblockN{Matvey Soloviev}
 \IEEEauthorblockA{\textit{KTH Royal Institute of Technology}}
\and \IEEEauthorblockN{Musard Balliu}\IEEEauthorblockA{\textit{KTH Royal Institute of Technology}}
 \and \IEEEauthorblockN{Roberto Guanciale}
\IEEEauthorblockA{\textit{KTH Royal Institute of Technology}}
}

\maketitle

\begin{abstract} 
We introduce a framework for reasoning about the security of computer systems
using modal logic. This framework is sufficiently expressive to capture
a variety of known security properties, while also being intuitive and
independent of syntactic details and enforcement mechanisms. We show
how to use our formalism to represent various progress- and termination-(in)sensitive
variants of confidentiality, integrity, robust declassification and
transparent endorsement, and prove equivalence to standard definitions.
The intuitive nature and closeness to semantic reality
of our approach allows us to make explicit
several hidden assumptions of these definitions, and identify potential
issues and subtleties with them, while also holding the promise
of formulating cleaner versions and future extension to entirely novel properties.

\end{abstract}

\section{Introduction}

The study of computer security is concerned with guaranteeing that computer
systems maintain some desirable properties in the face of attackers
that seek to subvert them by way of some set of actions and observations available to each of them.
To formalise and streamline this process, various  \emph{security properties}, 
ranging from basic ones such as confidentiality and
integrity \cite{goguen-meseguer,biba} to more involved ones such as robust declassification \cite{DBLP:conf/csfw/ZdancewicM01} and transparent endorsement \cite{cecchetti2017nonmalleable},
have been  studied. These properties can be used as
building blocks to provide extensional meaning to an overall \emph{security policy}, that is,  a description of desirable 
properties that we wish to be maintained against an attacker with well-defined capabilities. At the same time,
they serve as baseline for justifying and validating the soundness of enforcement mechanisms.

Generally, these security properties are defined with respect to particular system models and formalisms, which are
often tailor-made for the problem at hand.  In particular, constraints in specification techniques, programming language features, compositional reasoning, 
and details and limitations in the enforcement mechanisms have been intertwined in such a way that it has often been unclear exactly what security properties are enforced and how these properties relate to each other. Moreover,
the differences in these details often obscure subtleties
that determine their interpretation in counterintuitive scenarios, such
as ones involving nontermination, asynchrony or unusual sets of attacker capabilities  \cite{DBLP:conf/csfw/ZdancewicM01,DBLP:conf/csfw/MyersSZ04,DBLP:conf/csfw/ChongM06,DBLP:journals/jcs/MyersSZ06,DBLP:conf/pldi/BalliuM09,DBLP:journals/corr/abs-1107-5594,cecchetti2017nonmalleable,DBLP:conf/csfw/OakABS21}.

In this paper, we set out to explore security properties through the lens of modal logic  \cite{BRV01,rak}. We show that modal logic is a natural tool for expressing a range of security properties pertaining to intricate interactions of confidentiality and integrity in deterministic systems. 
We introduce a way of encoding the security-relevant behaviour of a system and the agents involved in its operation as \emph{security Kripke frames}, in which modal formulae can be evaluated. 
Drawing on a general interpretation of knowledge, capabilities and permissions as modal operators,
 we show how modal logic formulae can elegantly and intuitively capture security properties related to confidentiality and integrity along with notions of declassification (intended release of sensitive information) and endorsement (intended untrusted influence upon trusted data).
We then study the interplay between confidentiality and integrity by capturing security properties involving both, specifically robust declassification and transparent endorsement. We show each of these properties to be equivalent to their usual trace-based definition under the given conversion.
In the process, we have uncovered several subtleties of the established definitions, including a potential issue with how termination-insensitive robust declassification handles
termination and interesting details in how transparent endorsement deals with cause and effect.
Due to its generality, our framework can express these security properties for a range of attacker models, including passive, active, termination- and progress-(in)sensitive attackers, and system models, including synchronous, batch-job, and interactive systems.   Another advantage of our framework is that while programs with formal semantics can be converted into it straightforwardly, it also lends itself to manual representation of intuitive abstract scenarios that are not actually programming language-based, which can aid
in understanding various security properties. 

We are certainly not the first to point out the connection between modal logics and security properties. 
Since Sutherland's early work on non-deducibility \cite{Suth86}, several past works on information flow control have used the epistemic concept of knowledge as a fundamental mechanism to bring out what security property is being enforced \cite{dima2006nondeterministic,askarov2007gradual,DBLP:journals/corr/abs-1107-5594,DBLP:conf/nordsec/Balliu13,AhmadianB22,DBLP:conf/eurosp/McCallBJ22}. 
\changed{These works have produced elegant and intuitive security conditions for confidentiality and various flavours of declassification, but they do not address concerns related to the interplay between integrity and confidentiality.
One approach that considers both is the modal framework of Moore et al. \cite{mooreabstract}, which can capture many of the properties we discuss. However, their definition does not explicitly model the operational evolution of the system over time, which necessitates using the capabilities of agents as a proxy for system behaviour at the price of imposing additional restrictions on what policies can be expressed.
We discuss the connection to these works further in Section \ref{rw}.
}

The main contributions of this paper are the following:
\begin{itemize}
\item We introduce a way of capturing security-relevant information about a computer system, including each agent's capabilities and permissions for both passive observation and active interference, as Kripke frames, in which modal logic can be used to reason about its behaviour.
\item We show how to  represent a range of security properties, including confidentiality, integrity, declassification, endorsement, robust declassification, and transparent endorsement, as modal logic formulae in this framework. 
\item We demonstrate how these formulae can be adapted to  represent and compare versions of these properties for a range of attacker models, including passive, active, termination- and progress-(in)sensitive attackers, and system models, including synchronous, batch-job, and interactive systems\changed{, and explore potentially interesting novel variant definitions}.
\item We prove equivalence between our modal security properties and popular trace-based versions of the respective security properties in the literature, uncovering several subtleties of the established properties in the process.
\end{itemize}

\section{Background}

\subsection{A Primer on Modal Logic}
\label{sec:modallogic}

In mathematical logic, modifiers upon basic statements $\varphi$ such as ``it is possible that $\varphi$''
or ``$\varphi$ will eventually hold'' are referred to as \emph{modalities}. Logics which capture modalities and 
provide an interpretation for them are known as modal logics. Perhaps the most common example of a modal
logic is the propositional logic augmented with modalities for time and possibility, in which we write
$\Box \varphi$ for ``$\varphi$ is forever true'' (resp. ``necessarily''), and $\Diamond \varphi$ for ``$\varphi$
will be true at some point'' (resp. ``possibly''). This modal logic and others are customarily interpreted
using \emph{possible-world} (or \emph{Kripke}) \emph{semantics}, due to Saul Kripke \cite{BRV01,rak}. In this semantics,
we assign truth values to formulae in a \emph{Kripke frame}, which is a collection of 
\emph{possible worlds} $W$ and named relations $W\times W \supseteq R_1$, $R_2$, $\ldots$ on these worlds, as well as an interpretation $w\vDash \phi$ of
propositional formulae $\phi$ (i.e. ones that do not yet include any modalities) for every world $w\in W$.
If we seek to represent time, the possible worlds can be taken to be snapshots of the object of discourse at different times;
at each of them, propositional formulae talking about the state at that time can be evaluated and may take different truth values.
A single (reflexive, transitive) relation $T$ can then be introduced to relate pairs $(w_1,w_2)$ if $w_2$ lies
in the future of $w_1$.

The language of formulae is formed from the closure of the propositional formulae that we could interpret in the individual worlds
with additional modalities for each relation $R$, written as $[R] \varphi$ and $\langle R\rangle \varphi$ for any subformula $\varphi$.
In the case of the time relation $T$, these just correspond to the ``forever'' $\Box$ and ``eventually'' $\Diamond$
modality we mentioned above. We extend the interpretations $w\vDash$ to cover these modalities as follows:
\begin{align*}
& w\vDash [R]\varphi & \text{iff }w'\vDash \varphi\text{ $\forall w'$ s.t. $(w,w')\in R$} \\
& w\vDash \langle R\rangle \varphi & \text{iff $\exists w'$ with $(w,w')\in R$ and }w'\vDash \varphi
\end{align*}
It is easy to check that $w\vDash [R]\varphi$ iff $w\vDash \neg\langle R\rangle \neg \varphi$.

Apart from the time modality, we will use several other modalities to capture the capabilities 
and permissions of participants in a system over the course of this paper. Of particular interest will
be a modality of knowledge, with $\B{\KC_A}\varphi$ representing that $A$ knows that $\varphi$ is true,
and $\langle \KC_A\rangle \varphi$ representing that $A$ considers $\varphi$ possible.
As implied by the box-diamond duality, this modal-logic interpretation of knowledge is based on
an interpretation of ignorance (non-knowledge) of the truth of a formula $\varphi$ as
arising from the ignorant party being unsure as to the real state of the world they
find themselves in, and thus being unable to rule out the possibility that $\neg\varphi$ is true.
This way of modelling knowledge has been used successfully in a variety of contexts \cite{rak}, and also corresponds naturally to a view in security research \cite{askarov2007gradual,DBLP:journals/tissec/HalpernO08,DBLP:conf/pldi/BalliuDG11} where
limitations to knowledge derive from having a partial view of the system (and thus
not being able to rule out possible counterfactual states of parts of the system that one
may not see).

\subsection{Language-Based Security}
\def\Vars{V}
\def\Val{\mathsf{Val}}

Information flow control is a popular security framework for reasoning about 
dependencies between information sources and sinks, 
ensuring that these dependencies adhere to desirable (trace-based)
security properties. In a language-based setting, this security framework has the following components: 
(1) a program model which is given by the execution semantics of a
program; (2) an attacker model specifying the capabilities of an attacker over program executions (e.g., as read and write capabilities); (3) a permission model 
(or security policy) specifying the allowed or disallowed information flows in the program model; 
(4) a security property (or security condition) providing an end-to-end (extensional) characterisation of a program’s security with respect to these three components,
essentially ensuring that the attacker capabilities do not exceed the permissions.
A classical security property is noninterference \cite{goguen-meseguer}, which requires that any two executions starting with
equal values on public sources yield equal values on public sinks. This amounts to an attacker capable of only reading public sources and sinks, 
and a permission model which disallows information from secret sources to public sinks.

\vspace{0.5em}
\noindent \textbf{Program model.} Throughout this paper, we will discuss several (trace-based) security properties
of programs that occur in the literature. To describe these properties,
we define a model of programs written in a standard imperative \textsf{while}-language.
This model closely resembles those used in other security literature \cite{DBLP:journals/jsac/SabelfeldM03},
so we can use it to give the standard definitions;
and as we shall see, programs in this model can also be straightforwardly converted into
\textit{security Kripke frames}.
 Formally, 
\begin{definition} \label{def:program}
A program $P$
is a tuple $\langle p, V, I\rangle$, where 
\begin{itemize}
\item $p\in \mathsf{Programs}$ is the program text, drawn from a simple imperative $\mathsf{while}$ language.
\item $\Vars$ is the set of variables on which $p$ operates.
\item $I$ assigns to each variable $v$ a domain $\Val_v$ of values that it can take. For simplicity, we write $\Val \triangleq \bigcup_v  \Val_v$.  
\end{itemize}
\end{definition}
We assume that the program semantics are given by a \textbf{deterministic}
small-step semantical relation $(\rightarrow)$,
which relates \emph{configurations} (pairs of programs and \emph{stores}) in $(\mathsf{Programs}\times (\Vars\rightarrow \Val))$.
\def\Tr{\mathrm{Tr}}
\begin{definition} \label{def:trace} \begin{enumerate}[(i)]
\item For a program $P$, the transition relation $\rightarrow$ induces a set of \emph{possible traces} $\Tr(P)$.
This set contains, for each initial store
$\sigma_0: \Vars\rightarrow \Val$, where $\sigma_0(v)\in I(v)$ for all $v\in \Vars$,
all finite (partial) traces
$\langle p, \sigma_0\rangle \rightarrow \ldots \rightarrow \langle p_n, \sigma_n\rangle$.
\item When $t=\langle p, \sigma_0\rangle \rightarrow \ldots $ is either infinite or can not be extended by another transition $\rightarrow$, we refer to $t$ as \emph{maximal}. By determinism, every trace can be uniquely extended to a maximal trace.
\item If the maximal trace is finite, it must end in a halting state. In that case, we denote the unique maximal trace extending a given trace $t$ by $t_{\Downarrow}$.
\item When $t=\langle p, \sigma_0\rangle$, we refer to $t$ as \emph{initial} or \emph{singleton}. Every trace is the extension of a unique initial trace.
\end{enumerate}
\end{definition}

Note that infinite maximal traces are not contained in $\Tr(P)$, since they do not represent a state in which the system could actually be during its execution.

In examples throughout the paper, we will take programs to be given in a
standard imperative language with bounded and unbounded loops and mutable variables, but without functions.
To make the presentation frictionless, we assume this language to come
with appropriate operators to manipulate various data types, such as list concatenation $::$
and set union $\cup$.
\def\view{\mathsf{view}}
\def\fix{\mathsf{fix}}

\def\op{\mathsf{op}}

\vspace{0.5em}
\noindent \textbf{Capabilities and permissions.} In order to reason about the security of a program, we
define what actions (\emph{capabilities}) are available to the \emph{agents} participating in it,
and the properties we want to be preserved regardless of the actions any agent may take --
or, complementarily, what properties we want to \emph{permit} each agent to violate.
All the security properties from the literature we discuss in this paper are based on a particular
setup, where each agent's capabilities and permissions are defined in terms of being
allowed to read and write some subset of variables in $\Vars$.
We refer to the structure that contains this data for each agent $A$ as 
 the \emph{standard security context} (s.s.c.).

The read capabilities of the agent $A$ are described by a set of variables $R(A)$. Following the approach of Myers et al. \cite{DBLP:conf/csfw/MyersSZ04}, we take every agent to be able to observe whenever a value changes, but not to have an intrinsic ability to keep track of how much time passed between changes.
Each agent's read \emph{permissions} exactly coincide with their read \emph{capabilities}.
Moreover, each agent has the capability to influence the \emph{initial} state of another subset of variables $W(A)$,
which can be seen as attacker inputs to the program.
In terms of permissions, the agent is taken to be allowed to influence the contents of those same variables throughout the execution. 

Formally, for a program $P$ and set of agents $\mathcal A$, the s.s.c. includes maps $W, R: \mathcal A\rightarrow \mathcal{P}(\Vars)$.
These maps induce a projection of configurations to what each agent can read and must leave invariant when writing. Subsequent configurations in a trace might have the same projection for $A$, resulting in $A$ not being able to tell that a step occurred; to eliminate such repetition, we define a \emph{destuttering} operator $\mathsf{dest}(\vec c)$ as the longest subvector of $\vec c$ such that no two subsequent entries of it are equal, e.g. $\mathsf{dest}( \langle 0,0\rangle, \langle 0,0\rangle, \langle 1,0\rangle, \langle 1,0\rangle, \langle 0,0\rangle) = ( \langle 0,0\rangle, \langle 1,0\rangle, \langle 0,0\rangle)$. %
We then define the view $\view$ and write invariant $\fix$ for $A$ as follows:%
\begin{align*}
&  \view_A(\langle p, \sigma_0\rangle \rightarrow \ldots \rightarrow \langle p_n, \sigma_n\rangle) 
= \mathsf{dest} \langle \sigma_0|_{\scaleto{R(A)\mathstrut}{5pt}}, \ldots, \sigma_n|_{\scaleto{R(A)\mathstrut}{5pt}} \rangle \\
&  \fix_A(\langle p, \sigma_0\rangle \rightarrow \ldots \rightarrow \langle p_n, \sigma_n\rangle) 
= \mathsf{dest} \langle \sigma_0|_{\scaleto{\Vars \setminus W(A)\mathstrut}{5pt}}, \ldots, \sigma_n|_{\scaleto{\Vars \setminus W(A)\mathstrut}{5pt}} \rangle 
\end{align*}

We use $\sigma|_S$ to refer to the function $\sigma$ with domain restricted to $S$.
For example, in a program with variables $\{a,b,c,d\}$, we may have
$W(A)=\{c,d\}$, $R(A)=\{b,c\}$, so $A$ may write to $c$ and $d$, and read from $b$ and $c$.
Then a store $\sigma$ can be described as a 4-tuple $(a,b,c,d)$, and for example
\todo{RG: Do we want this case where wirte is not a subset of read?}
\begin{align*}
\view_A &(\langle c:=d; a:=b, (1,2,3,4) \rangle ) \\
 & \rightarrow \langle a:=b, (1,2,4,4) \rangle \\
 & \rightarrow \langle \mathsf{stop}, (2,2,4,4) \rangle) 
 &= ( (2,3), (2,4) ). \\
\fix_A &(\ditto[48pt] ) 
 &= ( (1,2), (2,2) ). 
\end{align*}
Note that in this example, reads are observed at steps 1 and 2,
while writes are observed at steps 1 and 3.

\section{Security framework}\label{sec:sf}

\subsection{Security Kripke Frames}

Security is typically defined in terms of the capabilities and permissions of a set of agents involved
in the operation of a system. We want to capture this data to reason about it. Both the capabilities and permissions
can be passive or active: an agent may observe something about the system's execution, or influence it by changing the
state. 
We represent systems subject to our security analysis as Kripke frames whose worlds correspond to possible
state of the system at a given point in time. These frames are equipped with a time
relation describing the \changed{evolution, and hence the operational behaviour,} of the system (assumed to be deterministic in this paper),
and a number of other relations representing the knowledge, capabilities and permissions
of the agents involved.

\begin{definition}A \emph{security Kripke frame} over a set of worlds $W$ and agents $A$ 
is a frame equipped with
\begin{itemize}
\item a time relation $T$, which is a union of total orders
on disjoint subsets of $W$;
\item for each agent $A$, 
\begin{itemize}
\item a pair of equivalence relations $\KC_A$ and $\KP_A$, which relates
two worlds if $A$ is, respectively, incapable and not permitted to distinguish them.
\item  a pair of transitive and reflexive relations $\WC_{A}$ and $\WP_{A}$, which relates two worlds $(w_1,w_2)$ if $A$ is, respectively, capable and permitted to move the system from state $w_{1}$ to $w_{2}$.
\end{itemize}
\end{itemize}
\end{definition}
We will generally work in settings in which $\KC_A=\KP_A$, and 
write both of them simply as $\K_A$ for compactness.

The $T$-connected components play a special role, as each of them
defines a \emph{deterministic} (by virtue of its totality on them:
for any two worlds, we can always say that one of them precedes the other, or they are identical) history
of how the system may evolve assuming no agent actions. %
\begin{definition}
A \emph{run} is an inclusion-maximal set of worlds $W$ on which $T$ is a total order.
\end{definition}

We assume that the worlds themselves come with a standard propositional logic
(that is, atomic propositions linked by $\wedge$, $\vee$, $\neg$ and $\Rightarrow$,
partially interdefineable)
of statements about the system's state in those worlds. This interpretation
can then be extended into a modal logic as described in Section \ref{sec:modallogic}.

\subsection{Security Kripke Frames from Programs}

We can define a conversion into our Kripke framework for any program and security context, though some useful
properties will hold of the resulting security Kripke frame when the security context is an s.s.c. in particular.

\begin{definition} \label{def:kripint}
Given a program $P$ with security context $S$, the \emph{Kripke interpretation} of $P$ and $S$ is
the security Kripke frame with
\begin{itemize}
\item possible worlds: %
The set of possible finite traces $\Tr(P)$, as those
represent states that the system may actually be in.
\item A time relation $T$ which relates two worlds if one is a (non-strict) prefix of the other.
\item For each agent $A$, $\KC_A=\KP_A$ relate two worlds $w^1$, $w^2$ if
$\view_A(w^1) = \view_A(w^2)$. 
\item For each agent $A$, $\WC_A$ relates two worlds $w^1$, $w^2$
if they are initial traces (i.e. length-1) and $\fix_A (w^1) = \fix_A (w^2)$.
Agents are only capable of changing the initial state of the system, and can not
interfere with it after execution commences.
\item For each agent $A$, the write-permission relation $\WP_A$ relates two worlds $w^1$, $w^2$ if
$\fix_A(w^1) = \fix_A(w^2)$. %
\end{itemize}
\end{definition}

In the Kripke interpretation of a program, we give an interpretation to atomic formulae which describe the state of the program at known points in the trace, as follows:
\begin{itemize}
\item $w\vDash v@\tau=i$ iff $w=(\langle p_0,\sigma_0\rangle, \ldots, \langle p_n,\sigma_n\rangle)$, $\tau\leq n$ and $\sigma_\tau(v)=i$.
This formula denotes that at least $\tau$ steps have been made, and $v$ had value $i$ at the $\tau$th step.
\item $w\vDash {\Downarrow}$ iff the program has already terminated in $w$, that is, $p_n=\mathsf{stop}$.
In particular, $w\vDash \Diamond{\Downarrow}$ means that the run in which $w$ lies will terminate, and $\Box\neg{\Downarrow}$ means that it diverges.
\end{itemize}
Atoms which are true at a world are also true at all of their $T$-successors. This property turns out to be desirable for more complex formulae as well, characterising those formulae that can be thought of as describing \emph{events} or \emph{memories} that once made can not be invalidated by additional information. 

\begin{definition}
A formula $\varphi$ is \emph{temporally sound (t.s.)} in a security Kripke frame if for all worlds $w$, 
$w\vDash \varphi$ iff $w\vDash \Box \varphi$.
\end{definition}
The intuition will be made precise when we define perfect recall (Def. \ref{def:pr}),
as we will see that it is only possible for those formulae that are temporally sound. 

While we focus on one particular attacker model, many other attacker models
can be directly reduced to it or modelled with small adaptations. We discuss
this in more detail in Remark \ref{rmk1}.

\subsection{Depicting Security Kripke Frames}

We will often illustrate (parts of) security Kripke frames as graphs, where the worlds are nodes
and some or all relations are depicted as edges or (for equivalence relations) partitions
of the worlds. We use black directed edges for $T$, wavy lines (blue and dashed for $\WP_A$, red for $\WC_A$ where relevant) for the write relations and circles for the $\KC_A$ equivalence classes, typically referred to as \emph{knowledge sets}. We may omit drawing edges and circles where transitivity, reflexivity etc. of the respective relation imply them, or the relation in question is not relevant to the discussion.

\begin{example} Consider the program $b:=a$, which
copies the value of one variable into another in one time step. Depending on the
associated s.s.c., the security Kripke frame representation of this program changes.
We can illustrate all types of relations by considering just two s.s.c.s with one
agent $A$ each: one in which $A$ can read $b$ but not $a$ (and not write anything),
and one in which $A$ can write $a$ but not $b$ (and read everything).
Assuming the possible initial configurations are such $b$ is always $0$ and $a$
is either $0$ or $1$, we obtain two security Kripke frames, depicted next to each other here
(worlds labelled $ab$):
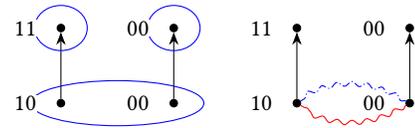
\begin{center}
\begin{tikzpicture}
    \runfinite{2}{1}{0}{0}
    \runfinite{2}{0}{1.5}{0}

     \draw [color=\perclr] (0.75,0) ellipse (1.15 and 0.3);
     \draw [color=\perclr] (0,1) ellipse (0.35 and 0.3);
     \draw [color=\perclr] (1.5,1) ellipse (0.35 and 0.3);

    \node [fill=white,inner sep=0.3pt] (l1) at (m00) {\small $00$ }; 
    \node [fill=white,inner sep=0.3pt] (l1) at (m10) {\small $10$ }; 

    \node [fill=white,inner sep=0.3pt] (l1) at (m01) {\small $00$ }; 
    \node [fill=white,inner sep=0.3pt] (l1) at (m11) {\small $11$ }; 

\end{tikzpicture}
\hspace{1em}
\begin{tikzpicture}
    \runfinite{2}{1}{0}{0}
    \runfinite{2}{0}{1.5}{0}

     \draw[color=\actclr, wave] (w00) to[bend left]  (w10);
     \draw[color=\perclr, wave, densely dashdotted] (w00) to[bend right]  (w10);

    \node [fill=white,inner sep=0.3pt] (l1) at (m00) {\small $00$ }; 
    \node [fill=white,inner sep=0.3pt] (l1) at (m10) {\small $10$ }; 

    \node [fill=white,inner sep=0.3pt] (l1) at (m01) {\small $00$ }; 
    \node [fill=white,inner sep=0.3pt] (l1) at (m11) {\small $11$ }; 
\end{tikzpicture}
\end{center}
It should be noted that each world actually contains more data than 
the labels imply: for example the world labelled $11$ at the top
left is really the trace $\langle b:=a, 10\rangle \rightarrow \langle (), 11\rangle$. 
We omit this data for concision, as the whole trace may be reconstrcuted from
looking at the run.

{%
\begin{remark}
\begin{enumerate}[(i)]
\item If we wish to model each agent as having the ability to track time, we can encode a clock variable as part of the transition semantics, and include it in $R(A)$. 
If all agents can observe the same clock variable (and none of them can write to it), we call the system \emph{synchronous}. 
  \item
        Some works instead choose to model interactive programs as having explicit FIFO output and/or input channels \cite{ONeillCC06}. Our model subsumes this, as we can represent the output channel of an agent by a special list-valued variable $O$, and define syntactic sugar $\mathsf{output}(e) \triangleq O:=(e::O)$.
Observations are generated exactly when this variable changes value, i.e. a new output is prepended. Similarly, programs with input channels can be modelled by considering an \textit{input stream}  $I$, and define $\mathsf{input}(x) \triangleq x:=\mathsf{head}(I); I:=\mathsf{tail}(I)$. Clark and Hunt  have shown that for deterministic interactive systems, streams are sufficient to model arbitrary interactive input strategies \cite{CH08}.
\item If we wish to model an agent not being able to observe a variable $v$ until some condition such as program termination is met, we can create a duplicate $v'$, put only $v'$ in $R(A)$ and have the semantics
force $v':=v$ iff the condition is met.
\item We sometimes encounter a seemingly more powerful attacker model, where the program contains some distinguished ``holes'' in which the attacker can insert code. This is actually equivalent to the model we use:
we can augment the program with fresh, attacker-writable variables whose contents encode the attacker code for each hole,
and fill the holes with copies of a special program that acts as an interpreter on the contents of the corresponding new variable.
\end{enumerate}\label{rmk1}
\end{remark}
}

\end{example}

\subsection{Properties of Security Kripke Frames}
Security Kripke frames can exhibit several natural properties which are useful in proving facts about their behaviour,
many of which are guaranteed to hold when a frame was obtained from an s.s.c. (and more generally
for many particularly well-behaved settings).

The first property is a straightforward consequence of the complete set of initial stores and projection-based
definitions of write and read permissions, but may hold in more general settings.
\begin{definition} In a Kripke frame, two relations $A$ and $B$ are said to \emph{commute}
if whenever $w_1 \sim_A w_2 \sim_B w_3$, there exists a $w_4$ such that
$w_1 \sim_B w_4 \sim_A w_3$, and likewise for $A$ and $B$ swapped.
\end{definition}
\begin{proposition} \label{prop:allcommute}
In a Kripke interpretation of a program $P$ with s.s.c. $S$,
for all agents $A$, $\WC_A$ commutes with $\KC_A$.
\end{proposition}

The next property simply describes that an agent may never forget something that it has once known. This may seem like an
imposition at first sight considering boundedness of memory, but in fact is usually a safe conservative assumption as 
situations in which an attacker actually can not muster enough memory to store a piece of information to execute an attack are rare.
\begin{definition} \label{def:pr} A security Kripke frame satisfies \emph{perfect recall} (for an agent $A$) if
for all t.s. $\varphi$,
$\B{\KC_A}\varphi \Rightarrow \Box \B{\KC_A} \varphi$.
\end{definition}
\begin{proposition} 
The Kripke interpretation of a program $P$ with s.s.c. $S$ satisfies perfect recall.
\end{proposition}

The restriction to temporally sound $\varphi$ is very nearly necessary, because given perfect recall, a formula that fails to be temporally sound at a world $w$ (so $w\vDash \neg\Box\varphi$) can not be known there, as 
 $\Box \B{\KC_A}\varphi$ implies $\Box\varphi$.

 The following property guarantees that the predicates are expressive enough to distinguish histories of the system. 
\begin{definition} A Kripke frame has \emph{characteristic formulae for prefixes}
if for every world $w$, there is a formula $\varphi_w$ that is only true at $w$ and its time successors $w'$ ($(w,w')\in T$).
\end{definition}
\begin{proposition}Security Kripke frames derived from s.s.c.s have characteristic formulae for prefixes.
\end{proposition}

Often, we assume that if a program has terminated, this becomes known
to every observer. Sometimes this assumption is explicit, but the more
common use of it is implicit in many older works of language-based security literature,
where the attacker is implicitly treated as receiving the program's output iff
the program in fact terminated, so the circumstance of seeing an output trace at all
conveys knowledge that the program has halted. Rather than encode a particular mechanism
by which this knowledge is conveyed in our representation of programs, we can formulate
a simple condition which encodes that this knowledge is made available by some mechanism
or another.

\begin{definition} \label{def:sigterm} A program $P$ \emph{signals termination}
if in its Kripke interpretation, for all agents $A$ and worlds $w$,
$$ w\vDash {\Downarrow} \Rightarrow \B{\KC_A}{\Downarrow}. $$
\end{definition}

\section{Security properties}

Having defined our security framework, we now show how it can be used to 
represent
a number of security properties as modal formulae which must be 
satisfied everywhere in the system. While our framework is general and extends to any system model that meets the assumptions
in Section \ref{sec:sf} (see Remark \ref{rmk1}), we show how it compares to popular trace-based security properties in the area of language-based security. As we will see,
the beauty of modal logic is that it naturally captures a range of intricate security properties pertaining to (complex interactions of) confidentiality and integrity,
and, at the same time, illuminates a number of issues that are not apparent from existing trace-based definitions of security properties.

In general, the security properties we propose all take a form that may be glossed as
$$w \vDash \text{potential problem} \Rightarrow \text{excuse}. $$
More specifically, we first identify a potentially problematic circumstance (for example, that
$A$ may eventually come to know a fact $\varphi$), and then describe a sufficient set
of conditions for this circumstance to not actually rise to the level
of a violation of the property (such as that $A$ may have known $\varphi$ from
the start, and hence did not just obtain it through the release of a secret, 
or that the system designer explicitly chose to approve the release of information).
Often, the ``problem'' will be some formula derived from the attacker's capabilities,
describing that the attacker can attain some effect or knowledge which might be questionable;
and the ``excuse'' will be a formula derived from the permissions,
describing that the effect or knowledge was in fact permitted under the security policy.

However, this is by no means a hard and fast rule, as we will find that
some security properties are naturally expressed in terms of problems being ``excused'' by other
capabilities. The problem itself may also involve permissions in some capacity,
such as when it recursively talks about some outcome not being permitted.
In general, one should keep in mind that there is no unique
way of writing any property in our framework, and we will therefore
choose formulations based on ease of explanation and expedience for equivalence proofs.

\subsection{Confidentiality}

Confidentiality is the property that no secret information is leaked to the public.
In a deterministic setting, if no new information is introduced into a (part of a) system,
then future states can be completely predicted from past ones. 
Conversely, if there are two states in which a system looks the same to an observer,
but proceeds to evolve differently, then the two states must have differed in a part
the observer could not see. If the observer then later gets to observe a difference
after all, he can rule out one or the other, and thereby learn something about the unseen
part. 
As in the security model underlying s.s.c.s, we conflate what an agent is capable of observing and what
the agent has the permission to know, though it would be easy to extend this
definition to situations where this is not the case (as we show for declassification in the appendix).

Per the above idea, the notion of confidentiality
therefore amounts to postulating that an agent's knowledge does not
increase over time, or more precisely that if an agent comes to know anything
in the future, then they must have known that this will be the case eventually
from the start. Formally, we make the following definition.
\begin{definition} \label{def:conf}
A security Kripke frame satisfies \emph{confidentiality} if
for all agents $A$, worlds $w$ and t.s. formulae $\varphi$,
$$w \vDash \Diamond \B{\KC_A}\varphi \Rightarrow \B{\KP_A}\Diamond \varphi. $$
\end{definition}
The prototypical example of a program that violates this definition is
the program $b:=a$ from the previous section, with $a$ secret and $b$ public:
in this program, at the initial world where $a$ equals $0$, the agent
will eventually come to know that $a$ was equal to $0$ ($\Diamond \B{\KC_A} (a@0=0)$), but is not allowed to 
know it initially ($\neg \B{\KP_A}\Diamond (a@0=0)$).
The same in fact holds at time $1$ for $b@1=0$, even though $a$ is allowed to know $b$, as this knowledge would also imply knowing the value of $a$.
\vspace{0.5em}
\noindent \textbf{Comparing to trace-based properties.}
Though it is impractically restrictive (as it does not
allow systems to have any public effect depending on secret data at all),
confidentiality is a frequently-invoked baseline security property,
which is often defined as a property of pairs of program traces. 
A simple definition in this style, which is equivalent to ours, can be given as follows.
\begin{definition} \label{def:confts}
(Adapted from \cite{DBLP:conf/csfw/MyersSZ04}, Def. 3)
A program $P$ and s.s.c. $S$
with read permissions $R$ satisfies \emph{trace confidentiality}
if for all agents $A$, initial stores $\sigma$
and pairs of assignments $v_1$ and $v_2$ to secret-to-$A$ variables
in $X=V(P)\setminus R(A)$,
if $t^i$ denotes the maximal trace generated from 
$\langle p,\sigma[X\mapsto v_i]\rangle$, then
 $\view_A(t^1)=\view_A(t^2)$.
\end{definition}

We establish the correctness of our definition by
showing that as we convert programs and s.s.c.s into
security Kripke frames by Definition \ref{def:kripint},
the resulting frame satisfies
the modal logic definition iff the original program satisfied
the trace-based definition.

\begin{theorem} \label{thm:confbasic} Suppose $P$ is a program with s.s.c. $S$.
Then $P$ and $S$ satisfy Def. \ref{def:confts}
iff the Kripke interpretation of $P$ and $S$ satisfies
Def. \ref{def:conf}.
\end{theorem}
\iffullver
\begin{proof} See \hyperlink{proof:confbasic}{appendix}. \end{proof}
\fi

Driven by the need to justify soundness of enforcement mechanisms, such as security type systems \cite{DBLP:journals/jcs/VolpanoIS96},
that do not account for program's termination, several authors have proposed
 \emph{termination-insensitive} notions of confidentiality: that
is, they choose to only track confidentiality violations involving halting
runs.
The motivation for this weaker attacker model can be seen as a combination of two assumptions:
firstly, the attacker is taken to not be
able to observe the running system, but instead only
to obtain a transcript of its observable part after it halts, and seeing
nothing at all if it does not. Secondly, any information that the attacker
could glean from observing whether the program terminates or not alone
is taken to be unproblematic by fiat (this is the part that motivates
the term ``termination-insensitive'')  \cite{DBLP:conf/esorics/AskarovHSS08}. 

It can be argued that it is more appropriate to encode
the circumstance that information only is revealed on termination
directly by setting up the $\KC$ relations to never
distinguish points that have not halted. We choose to not do this for
reasons of exposition, as our definition of the relations demonstrates
the use of the setup better and we can use it to define security
properties with different attacker models (e.g. progress-insensitivity)
 later on.
The compatibility with more involved properties such as robust
declassification also motivates us to adapt
the following definition from \cite{DBLP:conf/csfw/MyersSZ04}. 

\begin{definition} \label{def:conf2}
(Adapted from \cite{DBLP:conf/csfw/MyersSZ04}, Def. 3)
A program $P$ and s.s.c. $S$
with read permissions $R$ satisfies \emph{termination-insensitive trace confidentiality}
if for all agents $A$, initial stores $\sigma$
and pairs of assignments $v_1$ and $v_2$ to secret-to-$A$ variables
in $X=V(P)\setminus R(A)$,
if $t^i$ denotes the maximal trace generated from 
$\langle p,\sigma[X\mapsto v_i]\rangle$, then
either $\view_A(t^1)=\view_A(t^2)$, or
at least one of $t^1$ and $t^2$ does not halt.
\end{definition}

Due to the additional condition, this definition is of course not
equivalent to our previously stated modal one, unless the system
is such that all traces halt. Instead, we need to encode the additional
condition that both runs terminate.

\begin{definition} \label{def:confti}
A security Kripke frame satisfies \emph{termination-insensitive confidentiality} if
for all agents $A$, worlds $w$ and t.s. formulae $\varphi$,
$$w \vDash \Diamond{\Downarrow} \wedge \Diamond \B{\KC_A}\varphi \Rightarrow \B{\KP_A}(\Box\neg{\Downarrow} \vee \Diamond \varphi). $$
\end{definition}

We can read this as saying that if the run halts and $A$ eventually will learn $\varphi$, then
$A$ already knows that $\varphi$ whenever the run does not diverge.
To illustrate the differences between
the definitions, we compare their behaviour on various example programs.
\begin{example} In the following, we take variable $p$ to be readable to $A$ and $s$ to not be.
All violations, where they occur, involve the formula $\varphi \triangleq s@0=0$ at the length-1 world where $s=0$.
\begin{enumerate}[(i)]
\item The program $p:=s$ violates both \ref{def:conf} and \ref{def:confti}.
\item The program $p:=s; \mathsf{if}\,s=1\,\mathsf{then}\,\mathsf{loop}$
violates \ref{def:conf}, but satisfies \ref{def:confti}. It also violates the intermediate variant definition
using the formula $\Diamond{\Downarrow} \wedge \Diamond \B{\KC_A}\varphi \Rightarrow \B{\KP_A}\Diamond\varphi$, which could be seen as requiring that the program halt but \emph{not} permitting leaks that would already happen if the attacker could observe termination alone.
\item The program $p:=s; \mathsf{loop}$ violates \ref{def:conf}, but satisfies both \ref{def:confti} and the variant definition in (ii).
\end{enumerate}
\end{example}

As with the general case before, in
the termination-insensitive model, too, our definition is
equivalent to the trace-based one.

\begin{theorem} \label{thm:confeq} Suppose $P$ is a program with s.s.c. $S$.
Then $P$ and $S$ satisfy Def. \ref{def:conf2}
iff the Kripke interpretation of $P$ and $S$ satisfies
Def. \ref{def:confti}.
\end{theorem}
\iffullver
\begin{proof} See \hyperlink{proof:confeq}{appendix}. \end{proof}
\fi

Since the difference between the definitions only pertains to non-terminating runs,
it stands to reason that if all runs terminate, the two definitions in fact agree.
\begin{corollary} 
Suppose $P$ is a program with s.s.c. $S$ and all runs terminate.
Then $P$ and $S$ satisfy Def. \ref{def:conf2}
iff the Kripke interpretation of $P$ and $S$ satisfies
Def. \ref{def:conf}.
\end{corollary}
\begin{proof}
When $\Diamond {\Downarrow}$, the formulae of Defs. \ref{def:conf} 
and \ref{def:confti} are equivalent.
\end{proof}

{%
\begin{remark} %
An alternative way to define termination-insensitive confidentiality
would be to represent the circumstance that $A$ can only make observations
when the program terminates by encoding this in the capability relation $\KC_A$,
by saying that $w_1 \sim_{\KC_A} w_2$ iff either $\view_A w_1 = \view_A w_2$ and
both $w_1$ and $w_2$ represent halting states, or neither of $w_1$ and $w_2$ 
represent halting states.
In that case, Definition \ref{def:confti} could be written as 
$w\vDash \Diamond \B{\KC_A}\varphi \Rightarrow \B{\KP_A}(\Box\neg{\Downarrow} \vee \Diamond \varphi)$,
only needing to encode that knowledge was not gained by ruling out a
nonterminating run. We chose against this to avoid the confusing discussion of
multiple Kripke interpretations.
\end{remark}

\noindent \textbf{Progress insensitivity.} Another more permissive
notion of confidentiality that was suggested in Askarov et al. \cite{DBLP:conf/esorics/AskarovHSS08} is that of
\emph{progress-insensitive confidentiality}. As termination-insensitive
confidentiality stands for a notion of confidentiality ``modulo'' what
information is revealed by whether the program terminates or not,
progress-insensitive confidentiality is defined to ignore information
that could be gleaned from only observing how much progress the program made,
i.e. how many distinct states have been observed.
The rationale for this is
that data takes an exponential amount of time to leak through the progress
channel, and so for complex secrets, in certain settings everything
leaked through it can be treated as computationally infeasible to act upon.
An example of a program that satisfies progress-insensitive confidentiality
but not the other types we have seen is
$$ \mathsf{for}\,i=0,\ldots,s\,\mathsf{do}\,p:=i, $$
which produces a public-observable vector $(0,1,\ldots,s)$, where $s$
is the secret. 
This condition is enforced at every point in time, as
the attacker is taken to be able to observe the program in real time
as it executes.
 Hence, on the other hand, the program
$$ \mathsf{for}\,i=0,\ldots,s\,\mathsf{do}\,p:=s-i, $$ which produces
observations $(s,\ldots,0)$, does not satisfy it: after the first output the
public already knows that the secret is $s$ but would not from merely being able to
count.

The most straightforward way to capture this property is to introduce, for
each agent $A$, a special ``counting agent'' $A^\#$, whose knowledge represents
what can be learned by only observing the number of distinct outputs, so
$w^1 \sim_{\KC_{A^\#}} w^2$ iff $|\view_A(w^1)| = |\view_A(w^2)|$, that is, the
lengths match. Then we say that a ``problem'' only occurs if $A$ knows $\varphi$,
but $A^\#$ does not: we demand that for all $w$, $A$ and $\varphi$,
\begin{equation} \label{eqn:piconf}
 w\vDash \Diamond ( \B{\KC_A}\varphi \wedge \neg \B{\KC_{A^\#}} \varphi ) \Rightarrow \B{\KP_A} \Diamond \varphi. 
\end{equation}

\vspace{0.5em}
\noindent \textbf{Declassification.} If we wish to relax confidentiality by
identifying a particular formula $\psi$ whose truth value we want to allow 
the agent $A$ to learn even if it depends on secrets that $A$ can not observe initially,
we can achieve this by changing the permission relation $\KP_A$, setting
$w_1 \sim_{\KP_A} w_2$ iff $\view_A w_1 = \view_A w_2$ and either $w_1\vDash \psi$
and $w_2\vDash \psi$ or $w_1\not\vDash \psi$  and $w_2\not\vDash\psi$.
For example, if we want to declassify the XOR $s_1\oplus s_2$ of two secrets,
we obtain a security Kripke frame in which the world $(p\mapsto 0, s_1\mapsto 1, s_2\mapsto 1)$
is only $\KP_A$-related to worlds in which $s_1,s_2\mapsto 1$ or $s_1,s_2\mapsto 0$,
and so $\langle p:=s_1\oplus s_2, (0,1,1)\rangle \vDash \B{\KP_A} \Diamond (p@1=0)$.
This simple notion of declassification corresponds to \emph{what}-declassification
in the taxonomy of Sabelfeld and Sands \cite{SS09}.
}

\subsection{Integrity}
\label{sec:inte}
Integrity is the property that no trusted data is influenced by an untrusted party.
Often, this property is defined dually to confidentiality, with untrusted data
taking the place of secret data and trusted data taking the place of public data  \cite{biba}.
However, in our setting, it is more natural to model active attacker capabilities
as single modal steps, as opposed to the invariant-based representation we use
for attacker knowledge. This in turn suggests a formula which relates permissions
and capabilities in a different way than the one for confidentiality; however,
for s.s.c.-derived security Kripke frames, it can again be proven equivalent to the
standard definition, which implies that the two modal definitions are in fact
equivalent under additional conditions.

Integrity also differs from confidentiality in that \changed{the relations 
denoting the capabilities and permissions of the attacker must always be distinct.}
 This is because in the s.s.c.,
the attacker is assumed to only be able to influence the initial state of memory;
but some effects of this influence may linger for the duration of the system's execution,
without this constituting a \emph{prima facie} policy violation. 
For an agent $A$, the formula relates the capability relation $\WC_A$ to the
permission relation $\WP_A$. From a starting world $w$, the total set of possible states
that we could consider $A$ to be able to influence the system into taking are
those that are reached by $A$ performing an action available to it, followed by letting
the system run its course, i.e. the ones quantified over by $\langle \WC_A \rangle \Diamond$.
In some subset of those, $A$ will have influenced the system to bring about an effect
that $A$ was not actually permitted to bring about, as $A$'s influence was propagated
by the program code into a part of the system $A$ was not trusted with -- such as in the case
of an assignment $t:=u$ of an untrusted variable to a trusted one. 
In the parlance of the previous section, $\langle \WC_A \rangle \Diamond \varphi$ denotes
a potential problem, as $A$ caused the system to satisfy a property $\varphi$ which
$A$ may not actually have been allowed to.
The $\WP_A$ relation relates worlds to others
if the totality of differences between them is in aspects $A$ is allowed to affect. Since
we can assume that $A$ does nothing wrong by performing no action and merely waiting,
all effects attainable by $\Diamond$ followed by $\langle \WP_A \rangle$ can in fact be 
taken to be okay -- $\Diamond \langle \WP_A \rangle \varphi$ is a valid ``excuse'',
suggesting the following definition.

\begin{definition} \label{def:inte}
A security Kripke frame satisfies \emph{integrity} if for all agents $A$, worlds $w$
and t.s. formulae $\varphi$,
$$w \vDash \langle \WC_A\rangle \Diamond \varphi \Rightarrow \Diamond \langle \WP_A \rangle \varphi. $$
\end{definition}

This formula admits an intuitive reading saying that if $A$ can bring about $\varphi$
by exercising a capability and then waiting to let the system act on its behalf,
then it should be permitted to attain $\varphi$ directly after simply waiting for some amount of time.

We can also give a definition in 
the standard setting,
which is in fact structurally
identical to the trace-based definition (\ref{def:conf2}) of confidentiality. (Note that $\fix$
singles out the variables $A$ \emph{can not} modify, i.e. the ``trusted'' ones,
which correspond to the public ones under confidentiality-integrity duality.)

\begin{definition}  \label{def:inte2}
A program $P$ and s.s.c. $S$
with write permissions $W$ satisfies \emph{termination-insensitive trace-based integrity}
if for all agents $A$, initial stores $\sigma$
and pairs of assignments $v_1$ and $v_2$ to $A$-controlled variables
$X=W(A)$,
if $t^i$ denotes the maximal trace generated from 
$\langle p,\sigma[X\mapsto v_i]\rangle$, then
either $\fix_A(t^1)=\fix_A(t^2)$, or at least one
of $t^1$ and $t^2$ does not halt.
\end{definition}

As before, this is actually equivalent to a termination-insensitive counterpart of the
definition of integrity we just gave:

\begin{definition} \label{def:inteti}
A security Kripke frame satisfies \emph{termination-insensitive integrity} if for all agents $A$, worlds $w$
and t.s. formulae $\varphi$,
$$w \vDash \langle \WC_A\rangle (\Diamond{\Downarrow} \wedge \Diamond \varphi) \Rightarrow \Box\neg{\Downarrow}\vee \Diamond \langle \WP_A \rangle \varphi. $$
\end{definition}

We can read this property as saying that if by performing a write $A$
can make the system eventually halt and eventually make proposition $\varphi$
true, then either the system must currently diverge or $A$
must eventually be permitted to bring about $\varphi$ directly.

\begin{theorem} \label{thm:inteeq} Suppose $P$ is a program with s.s.c. $S$.
Then $P$ and $S$ satisfy Def. \ref{def:inte2}
iff the Kripke interpretation of $P$ and $S$ satisfies
Def. \ref{def:inteti}.
\end{theorem}
\iffullver
\begin{proof} See \hyperlink{proof:inteeq}{appendix}. \end{proof}
\fi

\begin{remark} The distinction between our attacker
model and one in which $A$ could influence variables in $W(A)$ at all times,
so $\WC_A=\WP_A$,
is illustrated by the program
$u:=0; t:=u$. This program first ensures that an untrusted
variable has value $0$ and then copies its value into a trusted variable. If the attacker
could change the value of $u$ after the first assignment, then the attacker could influence
the contents of $t$. In our attacker model, this satisfies integrity,
but it would not under the more powerful one.
\end{remark}
{%
\vspace{0.5em}
\noindent\textbf{Endorsement.}
As in the case of confidentiality,
we may want to relax the requirement that no untrusted agent may influence trusted
data at all.
Therefore, it may be of interest to allow for explicit \emph{endorsement}
of untrusted inputs using a syntactic construct.
 If we wish to endorse an agent $A$'s influence towards a piece
of trusted data, we may likewise choose to augment the
permission relation $\WP_A$.
Interestingly, influence and knowledge are structurally dissimilar, and so there
is no straightforward notion of endorsing influence along a particular formula.
We can, however, allow $A$ to influence the contents of a particular variable $v$,
setting $w^1\sim_{\WP_A} w^2$ iff $\fix_A w^1$ and $\fix_A w^2$ only differ in
entries corresponding to $v$.

Unlike in the case of declassification, we can even implement endorsement that
only allows effects after a particular endorsing event occurred in the execution
without changing our working definition. If we take the program to have a special
variable $E$ ranging over sets of pairs of agents and variable names, and define
a statement $\mathsf{endorse}(A,x) \triangleq E := E\cup \{(A,x)\}$, then we can
define an endorsement-aware permission relation as
$w^1\sim_{\WP_A} w^2$ iff $|f^1|=|f^2|$ and whenever $f^1 = \fix_A w^1$ and $f^2 = \fix_A w^2$ differ
in a variable $x$ their $i$th entry ($f^1_i(x) \neq f^2_i(x)$), influence to this
variable must already have been endorsed for $A$ in both ($x\in f^1_i(E), f^2_i(E)$).
A program such as $\mathsf{endorse}(A,t); t:=u$ (where $W(A)=\{u\}$) then satisfies integrity,
but $t:=u; \mathsf{endorse}(A,t)$ does not.
}

\subsection{Robust Declassification}
Since interesting programs rarely satisfy confidentiality, 
often, weaker security properties are considered.
One such property is robust declassification, which allows secrets
to be released, as long as \emph{whether} the secret is released can not be influenced by an untrusted party.
For example, the program $p:=s$, which always releases its secret, 
satisfies robust declassification;
 but
$\textsf{if}\,u=1\,\textsf{then}\,p:=s$ (Fig.~\ref{fig:violations}~(ii)) %
 (where $u$ is untrusted, $p$ public and $s$ secret) does not.

This property can be encoded almost straightforwardly based on our formula for confidentiality.
A more formal version can be stated as follows: If an attacker can
bring about a violation of confidentiality, then the violation must already exist
without the attacker doing anything. Since confidentiality is encoded by the implication
$ \Diamond \B{\KC_A}\varphi \Rightarrow \B{\KC_A}\Diamond \varphi$, a violation is just its negation,
$ \Diamond \B{\KC_A}\varphi \wedge \neg \B{\KP_A}\Diamond \varphi. $
With an additional condition we will explain shortly, we thus arrive at the following definition.
\begin{definition} \label{def:rd1}
A security Kripke frame satisfies \emph{robust declassification} iff
for all worlds $w\in \mathcal{F}$, all agents $A$ and all \emph{write-stable} t.s. formulae $\varphi$,
\begin{align*} w\vDash & \langle \WC_A \rangle ( \Diamond \B{\KC_A}\varphi \wedge \neg \B{\KP_A}\Diamond \varphi ) \\
& \Rightarrow \Diamond \B{\KC_A}\varphi \wedge \neg \B{\KP_A}\Diamond \varphi,%
\end{align*}
where a formula $\varphi$ is \emph{write-stable} for $A$ if for all worlds $w$,
$$ w\vDash \Diamond \varphi \Rightarrow [\WC_A] \Diamond \varphi. $$
\end{definition}
Why do we need the restriction to write-stable formulae?
If the restriction were not in place, we would encounter a problem:
 the $\WC_A$ step on the left-hand side would allow us to ``cheat'' the definition
by encoding differences which were actually directly brought about by the
attacker and do not just encode secrets whose release the attacker influenced.
For example, the program $p:=s$ with an additional unused untrusted variable $u\in W(A)$
would be declared insecure, as at the initial world $w$ where all variables have value 0,
the above formula is violated with $\varphi=(u@0=1) \wedge (s@0=0)$. Then the premise
of the implication holds, but the conclusion does not, as $w\vDash \Box\neg\varphi$ and hence
$w\vDash \neg \Diamond \B{\KC_A}\varphi$.  
More generally, suppose that $\varphi$ is in fact a formula which encodes a secret that gets released robustly (so
 $w\vDash \Diamond \B{\KC_A}\varphi \wedge \neg \B{\KP_A}\Diamond \varphi $ for all $w$),
and at a particular world $w$, $\psi$ is attacker-controlled in the sense that
 $w\vDash \langle \WC_A \rangle \Diamond\psi$ but $w\not\vDash \Diamond\psi$.
Then we have $w\vDash \langle \WC_A\rangle ( \Diamond \B{\KC_A}(\psi \wedge \varphi) \wedge \neg \B{\KP_A}\Diamond (\psi \wedge \varphi) )$,
but
$w\not\vDash \Diamond \B{\KC_A}(\psi \wedge \varphi) \wedge \neg \B{\KP_A}\Diamond (\psi \wedge \varphi)$,
since at $w$, $\psi\wedge\varphi$ is always false and therefore can not be known.
To avoid this problem, we thus assert that the eventual truth of the formula $\varphi$ must actually be independent
of $A$'s write permissions.
It is worth noting that when $\WC_A$ is an equivalence relation, as is the case
for Kripke interpretations of programs, write-stability also implies that
$$ w\vDash \neg \Diamond \varphi \Rightarrow [\WC_A] \neg \Diamond \varphi. $$

When $\KC=\KP$, the resulting definition
can in fact be simplified further: $w\vDash \langle \WC_A \rangle \neg \B{\KC_A}\Diamond \varphi$
implies by commutativity (Prop. \ref{prop:allcommute}) and write-stability that there
are also worlds where $\Box\neg\varphi$ that are $\KC_A$-related to $w$, and hence
the conjunct $\neg\B{\KC_A}\Diamond\varphi$ on the right-hand side becomes redundant.
A symmetric argument for when $\neg \B{\KC_A}\Diamond\varphi$, that is, the
conclusion is false, implies that the same conjunct is redundant on the left-hand side as well.

\begin{remark}
If $\KC_A=\KP_A$ for all $A$, the formula in Def. \ref{def:rd1} is equivalent to 
$$ w\vDash \D{\WC_A} \Diamond \B{\K_A} \varphi \Rightarrow \Diamond \B{\K_A}\varphi. $$
\end{remark}

Violations of robust declassification according to this definition
always take a particular appearance in the Kripke frame, presenting
as a pair of runs where $A$ gains knowledge (a knowledge set is refined over time into subsets) being related by
a possible write operation to a pair of runs where $A$ does not
(runs that were indistinguishable to $A$ remain such).
For the four runs where $p$ (which we overwrite anyway) is initialised to 0,
the example program from earlier (Fig.~\ref{fig:violations} (ii)) can be illustrated as follows.
Here, each world is labelled with a triple representing the state of $u$, $s$ and $p$ in it, in that order.
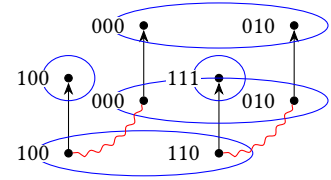
\begin{center}
\begin{tikzpicture}
    \runfinite{2}{2}{0}{0}
    \runfinite{2}{3}{2}{0}
    \runfinite{2}{0}{1}{0.7}
    \runfinite{2}{1}{3}{0.7}
     \draw [color=\perclr] (1,0) ellipse (1.45 and 0.3);
     \draw [color=\perclr] (2,0.7) ellipse (1.45 and 0.3);
     \draw [color=\perclr] (2,1.7) ellipse (1.45 and 0.3);
     \draw [color=\perclr] (0,1) ellipse (0.35 and 0.3);
     \draw [color=\perclr] (2,1) ellipse (0.35 and 0.3);

     \draw[color=\actclr, wave] (w00) to[bend left]  (w20);
     \draw[color=\actclr, wave] (w10) to[bend left]  (w30);

    \node [fill=white,inner sep=0.3pt] (l1) at (m00) {\small $000$ }; 
    \node [fill=white,inner sep=0.3pt] (l1) at (m10) {\small $010$ }; 
    \node [fill=white,inner sep=0.3pt] (l1) at (m20) {\small $100$ }; 
    \node [fill=white,inner sep=0.3pt] (l1) at (m30) {\small $110$ }; 

    \node [fill=white,inner sep=0.3pt] (l1) at (m01) {\small $000$ }; 
    \node [fill=white,inner sep=0.3pt] (l1) at (m11) {\small $010$ }; 
    \node [fill=white,inner sep=0.3pt] (l1) at (m21) {\small $100$ }; 
    \node [fill=white,inner sep=0.3pt] (l1) at (m31) {\small $111$ }; 

\end{tikzpicture}
\end{center}
The blue circles represent equivalence classes of the knowledge relations $\KP_A=\KC_A$,
and the red waves represent the non-reflexive component of the write capability relation $\WC_A$.
We can then see that the property is violated e.g. at the world with store
$000$: by taking a $\WC_A$-step to the world labelled $100$, we get a situation
where eventually (at the $100$-world) we have $\B{\KC_A}(s=0)$, but do not know that
this will be eventually the case (as $A$ thinks it may be at the $110$-world, where
eventually $\B{\KC_A}(s=1)$, instead).

\vspace{0.3em}
\noindent \textbf{Comparing to a trace-based definition.}
Once again, we want to compare our definition to a trace-based definition of robust declassification.
With slight adaptations to account for our notation and setting, the following definition is based on 
\cite{DBLP:conf/csfw/MyersSZ04}, but adapted using the updated notation and memory-based
attacker model of \cite{cecchetti2017nonmalleable} for consistency with our definitions.

\begin{definition} \label{def:rd2} (adapted from \cite{cecchetti2017nonmalleable}, Def. 6.5, and \cite{DBLP:conf/csfw/MyersSZ04}) A program $P$ and s.s.c. $S$ derived from $\langle W,R\rangle$ satisfies \emph{trace-based robust declassification}
if for all agents $A$, initial stores $\sigma$, pairs of assignments $v_1,v_2$ to secret-to-$A$ variables in $X=\Vars(P)\setminus R(A)$ and
pairs of assignments $w_1,w_2$ to $A$-writable variables in $Y=W(A)$, if $t^{ij}$ is the trace generated from $\langle p, \sigma[X\mapsto v_i][Y\mapsto w_j]\rangle$, then \emph{either}
$\view_A(t^{11})=\view_A(t^{21})$ iff $\view_A(t^{12})=\view_A(t^{22})$, \emph{or} at least one of the four runs diverges.
\end{definition}

As before, we need to encode the condition that all runs must terminate in the modal-logic formula in order to be able to prove equivalence for general programs. In the context of robust declassification,
this turns out to be somewhat less straightforward, as attacker actions can now influence
program termination independently of their effect on knowledge and truth of formulae pertaining
to other state. This can result in programs that satisfy Def. \ref{def:rd2} solely by virtue of, for instance, attacker actions always turning halting runs into non-halting runs and vice versa. This is especially a problem
in cases where the attacker may induce a disclosure of $\Diamond\varphi$ by making a previously always-halting program diverge whenever $\Box\neg\varphi$ (thus making it so that the program halting proves to the attacker that $\Diamond\varphi$ is true): by nature our formula would be inclined to consider this a violation of confidentiality, but the termination-insensitive reference definition does not. To forestall this scenario, we stipulate that 
the attacker must consider it impossible for the system to diverge if $\Box\neg\varphi$ after performing the attack in question; this may necessitate picking such $\varphi$ that their negation implies termination, at least given the particular attack. Figure \ref{fig:violations} and Remark \ref{def:rdvars} show what would happen without this additional conjunct.

\def\extrasecret{h}
\begin{figure*}
\begin{tabular}{rrcccc}
\toprule
& Program & Def. \ref{def:rd1} & Rmk. \ref{def:rdvars} (a) & Rmk. \ref{def:rdvars} (b) & Def. \ref{def:rdit} \\
\midrule
 (i) & $p:=s$               & -- & --  & -- & -- \\
(ii) & $\mathsf{if}\,u=1\,\mathsf{then}\,p:=s$ & $s@0=0$ & $s@0=0$ & $s@0=0$ & $s@0=0$ \\
(iii)& ($\mathsf{if}\,u=1\,\mathsf{then}\,p:=s);\, \mathsf{loop}$ & $s@0=0$ & -- & -- & -- \\
(iv)& $\mathsf{if}\,u=1\,\mathsf{then}\,(p:=s;\, \mathsf{if}\,s=1\,\mathsf{then}\,\mathsf{loop})$ & $s@0=0$ & $s@0=0$ & -- & -- \\
(v) & ($\mathsf{if}\,u=1\,\mathsf{then}\,p:=s);\, \mathsf{if}\,s\wedge (u\oplus \extrasecret)=1\,\mathsf{then}\,\mathsf{loop}$
& $s@0=0$ & $s@0=0$ & $s@0=0$ & -- \\
(vi) & ($\mathsf{if}\,u=1\,\mathsf{then}\,p:=s);\, \mathsf{if}\,(s\wedge \extrasecret)=1\,\mathsf{then}\,\mathsf{loop}$
& $s@0=0$ & $s@0=0$ & $s@0=0$ & $\extrasecret@0=1 \vee s@0=0$ \\
\bottomrule
\end{tabular}

\vspace{0.2em} All violations are at the length-1 world where all variables have value 0. The agent is allowed to read $p$ and write to $u$.
\caption{Violations of different formulae on various examples.}
\label{fig:violations}
\end{figure*}

\begin{definition} \label{def:rdit}
A security Kripke frame satisfies \emph{termination-insensitive robust declassification} iff
for all worlds $w\in \mathcal{F}$, all agents $A$ and all \emph{write-stable} t.s. formulae $\varphi$,
\begin{align*} & \langle \WC_A \rangle ( \Diamond{\Downarrow} \wedge \Diamond \B{\KC_A}\varphi \wedge \neg \B{\KP_A}(\Box\neg{\Downarrow} \vee \Diamond \varphi ) \\ 
 & \hspace{2em} \wedge \B{\KC_A} (\Box\neg\varphi \Rightarrow \Diamond{\Downarrow}) ) \\
& \Rightarrow \Box\neg{\Downarrow} \vee (\Diamond \B{\KC_A}\varphi
\wedge \neg \B{\KP_A}(\Box\neg{\Downarrow} \vee \Diamond \varphi ) ) .
\end{align*}
If $\KC_A=\KP_A$ for all $A$, this simplifies to
\begin{align*} & \langle \WC_A \rangle ( \Diamond{\Downarrow} \wedge \Diamond \B{\K_A}\varphi
\wedge \B{\K_A} (\Box\neg\varphi \Rightarrow \Diamond{\Downarrow}) ) \\
& \Rightarrow \Box\neg{\Downarrow} \vee \Diamond \B{\K_A}\varphi.
\end{align*}
\end{definition}

As discussed previously in the context of Def. \ref{def:sigterm}, classical language-based security setups implicitly assume that agents know if the program has terminated. This assumption is actually intimately linked with the justification behind the termination-insensitivity condition, and in this proof we will actually require it. With the requirement that the program signals termination made explicit, we can again prove equivalence.
We also require the assumption that the attacker can actually know what attack they performed, represented by $W(A)\subseteq R(A)$ (so everything that $A$ can write is also readable to $A$). This is implicit in the definition of \cite{DBLP:conf/csfw/MyersSZ04} and sensible in a perfect-recall context. 
\footnote{Uncertainty about the exact outcome of an attack can still be represented operationally, by having the attack affect a special piece of state representing that the attack occurred, and then operationally copying from a secret source of randomness.}
\begin{theorem} \label{thm:syncrd}
Suppose $P=\langle p, V, I\rangle $ signals termination and $S$ is an s.s.c. defined by $\langle W, R\rangle$, with $W(A)\subseteq R(A)$ for all agents $A$.
Then $P$ satisfies trace-based robust declassification (Def. \ref{def:rd2})
iff the Kripke interpretation of $P$ satisfies termination-insensitive robust declassification (Def. \ref{def:rdit}).
\end{theorem}
\iffullver
\begin{proof} See \hyperlink{proof:syncrd}{appendix}. \end{proof}
\fi
As in the case of our previous definitions, the additional conditions introduced in Def. \ref{def:rdit} are really
only necessary to deal with settings in which some runs may not terminate.
\begin{corollary} If all runs terminate, Def. \ref{def:rd1} and Def. \ref{def:rd2} are
equivalent under the same circumstances as in \ref{thm:syncrd}.
\end{corollary}

\vspace{0.5em}
\noindent \textbf{Reviewing the design space.} To better understand the components
of the formula in Def. \ref{def:rdit}, we consider whether and why it accepts or rejects programs
from a collection of examples, and the effect of omitting certain parts or
replacing them by alternative expressions. 

\begin{remark}[Wrong alternatives] 
\begin{enumerate}[(a)]
\item Omitting all components pertaining to termination of runs where $\Box\neg\varphi$: 
$$ \langle \WC_A \rangle ( \Diamond{\Downarrow} \wedge \Diamond \B{\K_A}\varphi  ) 
 \Rightarrow \Box\neg{\Downarrow} \vee \Diamond \B{\K_A}\varphi.$$
\item 
\changed{
Replacing the condition ($[K_A](\Box\neg\varphi\Rightarrow\Diamond\Downarrow)$)
that all possible runs after an attack 
where $\Box\neg\varphi$ must terminate with the weaker condition (inherited from termination-insensitive confidentiality)
that there must exist \emph{one} possible run that terminates and has $\Box\neg\varphi$:
\begin{align*} 
&  \langle \WC_A \rangle ( \Diamond{\Downarrow} \wedge \Diamond \B{\K_A}\varphi \wedge \neg \B{\K_A}(\Box\neg{\Downarrow} \vee \Diamond \varphi )) \\ & \Rightarrow \Box\neg{\Downarrow} \vee \Diamond \B{\K_A}\varphi. 
\end{align*}
}
\end{enumerate}
 \label{def:rdvars}
\end{remark}
In Fig. \ref{fig:violations}, we see how these alternative attempts at a definition differ
in their behaviour from Def. \ref{def:rdit} and thus fall short of matching the trace-based definition.
All examples apart from (i) leak the secret $s$ only if the untrusted input $u$ equals 1 and
differ only in whether they subsequently terminate, so they all violate normal robust declassification,
Def. \ref{def:rd1}. For the other examples:
\begin{enumerate}
\item[(iii)] As this example never terminates, it naturally satisfies all formulae that encode
the assumption nothing is learned unless the program terminates.
\item[(iv)] This example loops forever if the attack was performed ($u=1$) and the secret was equal to 1.
Per termination insensitivity, it should be considered secure, as 
the information about $s$ that is leaked after performing the attack could be gleaned from just
observing whether the program terminated. However, the definition of \ref{def:rdvars} (a) is
violated, as the program terminates and reveals the secret when $s=0$, and the definition
does not concern itself with halting in the counterfactual scenario.
\item[(v)] When $s=1$, this example might or might not diverge depending on an additional
secret parameter $\extrasecret$, but also flips between terminating and not terminating depending on whether
the attack was performed. This still satisfies termination-insensitive robust declassification
as stated, because we can not find a violating quadruple of terminating runs.
A possible interpretation is that if performing the attack $u$ flips termination
status, we know that $s=1$, and hence the termination channel already reveals this information.
However, \ref{def:rdvars} (b) is violated, as it merely stipulates that after performing
the attack it is still \emph{possible} that the system halts with the truth of $\varphi$ going either way.
\item[(vi)] Here, an unrelated secret parameter $\extrasecret$ 
could lead to nontermination,
independent of the attack. This is correctly flagged as violating by all definitions, but illustrates
the cost of the $\B{\K_A} (\Box\neg\varphi \Rightarrow \Diamond{\Downarrow}) )$ condition, 
as it necessitates the violation $\varphi$ to include an additional framing disjunct
that excludes the unrelated nonterminating runs from the counterexamples.
\end{enumerate}

A more interesting question that the shape of Definition \ref{def:rdit} raises is why we
did not simply construct the formula in the most straightforward way by using 
termination-insensitive confidentiality on both the problem and the excuse side.

\begin{remark}[A more justifiable alternative] \label{def:rdjust}
By using termination-insensitive confidentiality on both the premise and conclusion side,
we get the formula
\begin{align*}
 & \langle \WC_A \rangle ( \Diamond{\Downarrow} \wedge \Diamond \B{\KC_A}\varphi \wedge \neg \B{\KP_A}(\Box\neg{\Downarrow} \vee \Diamond \varphi )) \\
 & \Rightarrow \Diamond{\Downarrow} \wedge \Diamond \B{\KC_A}\varphi \wedge \neg \B{\KP_A}(\Box\neg{\Downarrow} \vee \Diamond \varphi ).
\end{align*}
This principially differs from the previous family of definitions in that the ``excuse'' side 
has $\Diamond{\Downarrow} \wedge \ldots$ rather than $\Box\neg{\Downarrow}\vee\ldots$.
\end{remark}
The reason for not choosing this or a related definition is that it does not in fact agree with Def. \ref{def:rd2}
in its treatment of nonterminating runs, inherited from \cite{DBLP:conf/csfw/MyersSZ04}.
A distinguishing example is given by the program
$$ \mathsf{if}\,u=1\,\mathsf{then}\,p:=s\,\mathsf{else}\,\mathsf{loop}. $$
This program never terminates when $u=0$, but terminates and leaks the secret when the
untrusted input has been set to $u=1$.
Defs. \ref{def:rdit} and \ref{def:rd2}, and also the two variants in Rmk. \ref{def:rdvars},
would treat this program as secure, as all candidate quadruples of runs
include some non-terminating ones. 
Should it in fact be? We might
argue that according to the intended readings of termination insensitivity and robust declassification,
it should not: termination insensitivity states that we should ignore any leaks
that already would occur from merely observing the termination channel, but here when $u=0$ the
program never terminates, so the termination channel (or any other channel) does
not leak the secret which is straightforwardly released when $u=1$.
If one agrees with this reasoning, this would point at a potential problem with the Definition
of \cite{DBLP:conf/csfw/MyersSZ04}, which we had to exert some effort to reproduce faithfully.

\subsection{Transparent Endorsement}
Cecchetti et al. \cite{cecchetti2017nonmalleable} propose a property that
is dual to robust declassification, which they call \emph{transparent endorsement}. 
Analogously to how robust declassification relaxes confidentiality, this property amounts to a relaxation
of integrity. In this section, we will work towards
capturing this property in our framework (Def. \ref{def:teit}).
As this definition involves choices that may seem surprising without context,
we opt to proceed step by step, showing why more obvious approaches fail due to subtleties
in our integrity definition.

 Under transparent endorsement, untrusted
agents are allowed to influence trusted aspects of the system;
however, \emph{whether} they influence them should not depend on an
input that is secret to the untrusted agent in question.
We can illustrate three cornerstones of this definition as follows:
\begin{example} \label{ex:te} \begin{enumerate}[(i)]
\item A simple assignment from untrusted to trusted,
$t:=u$, satisfies transparent endorsement.
\item An assignment that occurs on a secret condition, 
$\textsf{if}\,s=1\,\textsf{then}\,t:=u$ (with $s$ secret, $t$ trusted and $u$
untrusted), violates transparent endorsement.
\item If the secret condition only determines \emph{how} the
integrity violation occurs, as in $ \mathsf{if}\,s=1\,\mathsf{then}\,t_1:=u\,\mathsf{else}\,t_2:=u, $
transparent endorsement is satisfied.
\end{enumerate}
\end{example}
\begin{remark} A motivation for this definition is given \changed{(\cite{cecchetti2017nonmalleable}, Fig. 1, 2)} by settings in which a user input
such as a password or a bid in an auction is endorsed to a higher integrity level
for further processing, in the assumption that the input really came from the user
in question. If the user has the ability to pass in inputs that the user himself can not
see, e.g. because they are encrypted, this assumption becomes questionable, leading to
unexpected outcomes: for instance, an attacker could submit an encrypted  
password as input to a password checker without any issues being flagged (as 
the password was never leaked), or cheat an auction with a malleable
encryption scheme by a specifically prepared bid that decodes to the opponent's bid
incremented by 1. Since the language-based constructs used to implement those examples
in full go beyond the scope of this paper, here we only consider the property in isolation.
\end{remark}

If $s$, $t$ and $u$ all range over $\{0,1\}$ and all possible combinations
are possible as initial stores, we can represent
the 8-run system of (ii) as follows
(nodes labelled with the values of $s$, $t$ and $u$ in order):
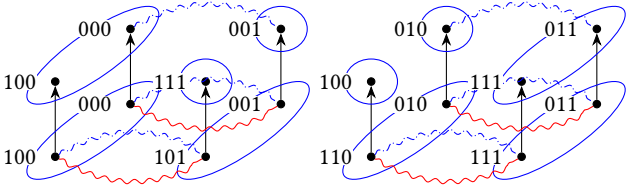
\begin{center}
\begin{tikzpicture}
\foreach \t/\offs in {0/0,1/4.2} {
\begin{scope}[shift={(\offs,0)}]
    \runfinite{2}{2}{0}{0}
    \runfinite{2}{3}{2}{0}
    \runfinite{2}{0}{1}{0.7}
    \runfinite{2}{1}{3}{0.7}

\ifthenelse{\t=0}{
     \draw [color=\perclr] (w11) ellipse (0.35 and 0.3);
     \draw [color=\perclr] (w31) ellipse (0.35 and 0.3);
     \draw [color=\perclr,rotate around={35:(0.5,1.35)}] (0.5,1.35) ellipse (1.05 and 0.3);
}{
     \draw [color=\perclr] (w01) ellipse (0.35 and 0.3);
     \draw [color=\perclr] (w21) ellipse (0.35 and 0.3);
     \draw [color=\perclr,rotate around={35:(2.5,1.35)}] (2.5,1.35) ellipse (1.05 and 0.3);
}
     \draw [color=\perclr,rotate around={35:(0.5,0.35)}] (0.5,0.35) ellipse (1.05 and 0.3);
     \draw [color=\perclr,rotate around={35:(2.5,0.35)}] (2.5,0.35) ellipse (1.05 and 0.3);

     \draw[color=\actclr, wave] (w00) to[bend right]  (w10);
     \draw[color=\actclr, wave] (w20) to[bend right]  (w30);
     \draw[color=\perclr, wave, densely dashdotted] (w01) to[bend left]  (w11);
     \draw[color=\perclr, wave, densely dashdotted] (w00) to[bend left]  (w10);
     \draw[color=\perclr, wave, densely dashdotted] (w20) to[bend left]  (w30);

    \node [fill=white,inner sep=0.3pt] (l1) at (m00) {\small $0\t0$ }; 
    \node [fill=white,inner sep=0.3pt] (l1) at (m10) {\small $0\t1$ }; 
    \node [fill=white,inner sep=0.3pt] (l1) at (m20) {\small $1\t0$ }; 
    \node [fill=white,inner sep=0.3pt] (l1) at (m30) {\small $1\t1$ }; 

    \node [fill=white,inner sep=0.3pt] (l1) at (m01) {\small $0\t0$ }; 
    \node [fill=white,inner sep=0.3pt] (l1) at (m11) {\small $0\t1$ }; 
    \node [fill=white,inner sep=0.3pt] (l1) at (m21) {\small $100$ }; 
    \node [fill=white,inner sep=0.3pt] (l1) at (m31) {\small $111$ }; 
\end{scope}
}
\end{tikzpicture}
\end{center}
Looking at the wavy lines that represent write permissions and capabilities,
 we note a dual of the pattern observed for r.d. before:
the initial points of two pairs of runs are related by $K_A$, but only
one of them remains related by $\WP_A$. %

Adapted to deal with non-termination in the same way as the reference
definition of robust declassification, the trace-based definition,
given in \cite{cecchetti2017nonmalleable} in the context of programs that
are guaranteed to terminate, can be stated as follows.

\begin{definition} \label{def:te2} (adapted from \cite{cecchetti2017nonmalleable}, Def. 6.7) A program $P$ and s.s.c. $S$ derived from $\langle W,R\rangle$ satisfies \emph{trace-based transparent endorsement}
if for all agents $A$, initial stores $\sigma$, pairs of assignments $w_1,w_2$ to secret-to-$A$ variables in $Y=\Vars(P)\setminus R(A)$ and
pairs of assignments $v_1,v_2$ to $A$-writable variables in $X=W(A)$, if $t^{ij}$ is the trace generated from $\langle p, \sigma[X\mapsto v_i][Y\mapsto w_j]\rangle$, then \emph{either}
$\fix_A(t^{11})=\fix_A(t^{21})$ iff $\fix_A(t^{12})=\fix_A(t^{22})$, \emph{or} at least one of the four runs diverges.
\end{definition}

It would appear natural to follow a similar approach as with robust
declassification and work off of the formula for integrity in Def. \ref{def:inte}
to express that if it is violated, this must be independent of secrets, perhaps 
quantifying
\begin{align} & \langle \WC_A\rangle \Diamond \varphi \wedge \neg \Diamond \langle \WP_A \rangle \varphi \nonumber \\
& \Rightarrow \B{\KC_A} ( \langle \WC_A\rangle \Diamond \varphi \wedge \neg \Diamond \langle \WP_A \rangle \varphi) \label{eqn:badte}
\end{align}
over an appropriate notion of \emph{read-stable} formulae that are independent of secrets,
that is, satisfy
$$ w\vDash \Diamond\varphi \Rightarrow [\KC_A]\Diamond\varphi $$
for all $w$.
However, this turns out to not work.
In Example \ref{ex:te} (ii), the violation that occurs when $s=1$ is
due to the formula $t@1=1$, that is to say, the violation encodes
that the effect that the trusted variable $t$ equals 1 could be brought about,
without this ever being directly permissible. But this formula is not in fact
read-stable, as at the initial world where $s$ and $t$ are both 1,
it will eventually be true but $A$ considers it possible that $s=0$ and
hence it would not be. Equivalently, by observing $t@1=1$, $A$ in fact learns
that the secret $s$ equals 1.
We might try to fix this by picking a formula which is carefully designed to
be read-stable, such as $\varphi\triangleq t@1=1 \vee (s@0=0 \wedge u@0=1)$,
and this indeed violates (\ref{eqn:badte}).
However, this trick turns out to be too powerful,
as the ostensibly safe example $t:=u$ violates (\ref{eqn:badte}) with
this formula at the all-0 initial world as well: $\D{\WP_A} \varphi$ is true as the ($s=0,u=1$)-world is $\WP_A$-reachable, and hence $\Diamond \D{\WP_A} \varphi$ is true, but the $s=1,u=0$ world at
which an integrity violation occurs is still $\KC_A$-reachable too.

The key to solving this problem instead lies in the difference between the structure
of the integrity formula (Def. \ref{def:inte}) and the confidentiality formula (Def. \ref{def:conf}).
Though we have proved them equivalent to structurally identical trace-based definitions,
the conditions that have to hold \emph{for each formula $\varphi$} are not in fact logically equivalent,
and the violating $\varphi$
in the two cases differ. In the case of integrity, which prohibits flows from untrusted
to trusted variables, the violations of our definition pertain to the trusted effect at the end of the flow;
on the other hand, for confidentiality, which prohibits flows from secret to public,
the violations we found instead pertained to the secret origins/``causes'' such as $s@0=0$.
The definition of transparent endorsement that we seek to match treats causes and effects
in a subtly different manner, as it asserts that secrets should not influence whether
a \emph{particular} untrusted input affects \emph{any} trusted effect. 
Hence, Ex. \ref{ex:te} (iii) should satisfy transparent endorsement.
This existential quantification over effects turns out to be hard to mirror by
restricting the quantification over formulae.
Even with the stipulation that a violating formula must be read-stable,
we can construct a violation of Eqn. \ref{eqn:badte} as
$\varphi\triangleq t_1@1=1 \vee (s@0=0 \wedge u@0=1)$,
which gives an integrity violation at the world where $s=1$ and $u=0$,
but not at the one where $s=0$ and $u=0$.
To reproduce the behaviour of robust declassification, we therefore define the following alternative
variant of integrity that instead structurally follows the formula for confidentiality,
except with permissions and capabilities swapped.

\begin{remark}[Cause form of integrity] \label{def:intecf} 
A security Kripke frame satisfies \emph{cause integrity} if for all agents $A$, worlds $w$
and t.s. formulae $\varphi$,
$$w \vDash \Diamond \B{\WP_A} \varphi \Rightarrow \B{\WC_A} \Diamond \varphi. $$
It satisfies \emph{cause termination-insensitive integrity} if for the same,
$$w \vDash \Diamond{\Downarrow} \wedge \Diamond \B{\WP_A}\varphi \Rightarrow \B{\WC_A}(\Box\neg{\Downarrow} \vee \Diamond \varphi). $$
\end{remark}

A possible natural-language reading of this definition is that
if eventually $A$ is not permitted to falsify $\varphi$, then $A$
must not be capable of performing any action that would result in
$\varphi$ being perpetually false. The program $t:=u$ violates
it with the formula $u@0=0$: once the value of $u$ has been
recorded in the trusted location $t$, $A$ is no longer allowed
changes that would result in a state only consistent with $u$ being 1,
but $A$ would be capable of performing such a change in the beginning.
We can then finally formulate transparent endorsement.

\begin{definition} \label{def:te}
A security frame satisfies \emph{transparent endorsement} iff
for all worlds $w\in \mathcal{F}$, all agents $A$ and all \emph{read-stable} t.s. formulae $\varphi$,
\begin{align*} w\vDash & \langle \KC_A \rangle ( \Diamond \B{\WP_A}\varphi \wedge \neg \B{\WC_A}\Diamond \varphi ) \\
& \Rightarrow \Diamond \B{\WP_A}\varphi \wedge \neg \B{\WC_A}\Diamond \varphi.%
\end{align*}
\end{definition}

\begin{definition} \label{def:teit}
A security Kripke frame satisfies \emph{termination-insensitive transparent endorsement} iff
for all worlds $w\in \mathcal{F}$,  agents $A$ and \emph{read-stable} t.s. formulae $\varphi$,
\begin{align*} w\vDash & \langle \KC_A \rangle ( \Diamond{\Downarrow} \wedge \Diamond \B{\WP_A}\varphi \wedge \neg \B{\WC_A}(\Box\neg{\Downarrow} \vee \Diamond \varphi ) \\ 
 & \hspace{2em} \wedge \B{\WP_A} (\Box\neg\varphi \Rightarrow \Diamond{\Downarrow}) ) \\
& \Rightarrow \Box\neg{\Downarrow} \vee (\Diamond \B{\WP_A}\varphi
\wedge \neg \B{\WC_A}(\Box\neg{\Downarrow} \vee \Diamond \varphi ) ) .
\end{align*}
\end{definition}

Since these formulae are structurally identical to the ones for confidentiality,
proofs of equivalence to the trace-based ones can proceed by symmetry.
\begin{proposition} \begin{enumerate}[(i)]
\item If $P$ is a program with s.s.c. $S$, then the Kripke interpretation of $P$ and $S$ satisfies
cause termination-insensitive integrity iff $P$ and $S$ satisfy trace-based integrity (Def. \ref{def:inte2}).
\item If $P$ also signals termination and $W(A)\subseteq R(A)$, then it satisfies trace-based transparent endorsement (Def. \ref{def:te2})
iff the Kripke interpretation of $P$ and $S$ satisfies termination-insensitive transparent endorsement (\ref{def:teit}).
\end{enumerate} \label{prop:teeq} 
\end{proposition}
\iffullver
\begin{proof} See \hyperlink{proof:teeq}{appendix}. \end{proof}
\fi

\begin{remark} Though Def. \ref{def:teit} is the one that matches the original definition
of transparent endorsement, a definition
which is violated by Ex. \ref{ex:te} (iii) may in fact be of interest. 
As it stands, it is not necessarily safe to reuse ``trusted'' outputs of a program
satisfying transparent endorsement: if a program
fragment such as that example is followed by another which elevates the contents
of only one of the two variables to a yet higher level of trust (say, $t^* = t_1$),
then two ``innocuous'' program fragments that individually satisfy transparent endorsement would
combine into a program that does not. A definition based on Eqn. \ref{eqn:badte} would avoid
this issue.
\end{remark}

{%
\subsection{Comparing Security Properties}
One advantage of our representation of security properties as modal implications
that have to hold for all subformulae $\varphi$ is that it is straightforward
to contextualise properties in terms of their mutual implications. A
property of the form $A\Rightarrow B$ implies, and thus is stronger than, another property $A'\Rightarrow B'$
if $A'\Rightarrow A$ and $B\Rightarrow B'$. 
For the properties we discussed, it usually turns out to be sufficient to consider
strengthenings and weakenings due to additional conjunctions and disjunctions respectively:
for example, when $\KC_A=\KP_A$, the simplified definition
of robust declassification (r.d.) (Def. \ref{def:rd1})
is $\langle \WC_A \rangle \Diamond \B{\K_A}\varphi \Rightarrow \Diamond \B{\K_A}$,
and since $\Diamond{\Downarrow} \wedge \Diamond \B{\K_A}\varphi \Rightarrow \Diamond \B{\K_A}\varphi$
and $\Diamond \B{\K_A} \varphi \Rightarrow \Box\neg{\Downarrow} \vee \Diamond \B{\K_A}\varphi$,
it implies the wrong attempt of \ref{def:rdvars}(a) to define
a termination-insensitive (t-e) version, 
$ \langle \WC_A \rangle ( \Diamond{\Downarrow} \wedge \Diamond \B{\K_A}\varphi  ) 
 \Rightarrow \Box\neg{\Downarrow} \vee \Diamond \B{\K_A}\varphi$,
 which we only formulated for that case. By using similar 
reasoning, we can construct 
the following overview of how the confidentiality properties discussed in this paper 
relate to each other.
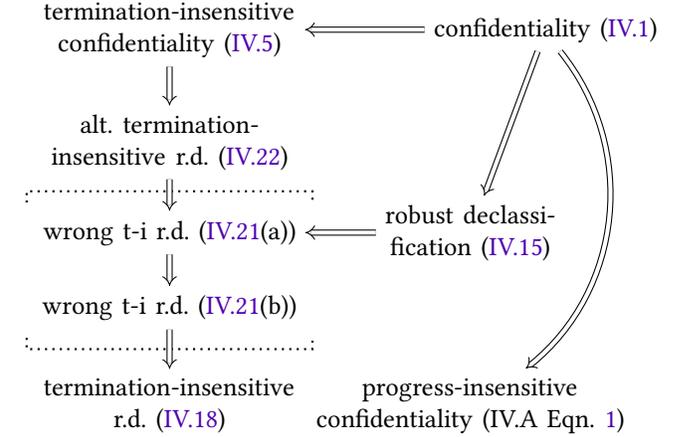
\begin{center}
\begin{tikzpicture}[>=latex]
    \node (conf) at (5,0) {confidentiality (\ref{def:conf})};
    \node[align=center] (confpi) at (4,-5) {progress-insensitive\\ confidentiality (IV.A Eqn. \ref{eqn:piconf})};
    \node[align=center] (confti) at (0,0) {termination-insensitive\\ confidentiality (\ref{def:confti})};
    \node[align=center] (rd) at (4,-2.7) {robust declassi-\\fication (\ref{def:rd1})};
    \node[align=center] (rdjust) at (0,-1.5) {alt. termination-\\insensitive r.d. (\ref{def:rdjust})};
    \node (rdalt1) at (0,-2.7) {wrong t-i r.d. (\ref{def:rdvars}(a))};
    \node (rdalt2) at (0,-3.7) {wrong t-i r.d. (\ref{def:rdvars}(b))};
    \node[align=center] (rdti) at (0,-5) {termination-insensitive\\ r.d. (\ref{def:rdit})};

    \draw [dotted, thick] (-1.9,-2.3) to [square left brace]  (1.9,-2.3);
    \draw [dotted, thick] (-1.9,-4.1) to [square right brace]  (1.9,-4.1);

    \draw[double equal sign distance,-implies] (conf)--(confti);
    \draw[double equal sign distance,-implies] (conf) to [bend left=45] (confpi);
    \draw[double equal sign distance,-implies] (conf)--(rd);
    \draw[double equal sign distance,-implies] (rd)--(rdalt1);
    \draw[double equal sign distance,-implies] (confti)--(rdjust);
    \draw[double equal sign distance,-implies] (rdjust)--(rdalt1);
    \draw[double equal sign distance,-implies] (rdalt1)--(rdalt2);
    \draw[double equal sign distance,-implies] (rdalt2)--(rdti);

\end{tikzpicture}
\end{center}
No additional implications beyond those implied by transitivity
hold among these properties in general (see Fig. \ref{fig:violations}
for some of the examples witnessing this). In particular, we note
that robust declassification implies the benchmark definition
of termination-insensitive robust declassification (Def. \ref{def:rdit})
but not our proposed alternative (Def. \ref{def:rdjust}).
A program witnessing this, $ p:=s;\, \mathsf{if}\,u=0\,\mathsf{then}\,\mathsf{loop} $,
can be identified straightforwardly by inspecting the formula, which
suggests that we fail the additional condition $\Diamond{\Downarrow}$
on the right-hand side.

We can establish a similar diagram for the integrity properties, which
are fewer in number:
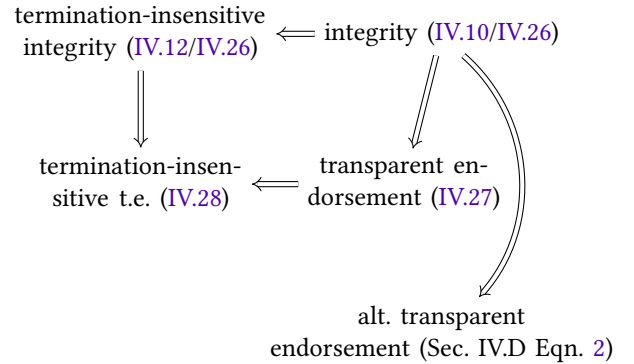
\begin{center}
\begin{tikzpicture}[>=latex]
    \node (inte) at (4,0) {integrity (\ref{def:inte}/\ref{def:intecf})};
    \node[align=center] (intewrong) at (4,-4) {alt. transparent\\ endorsement (Sec. IV.D Eqn. \ref{eqn:badte})};
    \node[align=center] (inteti) at (0,0) {termination-insensitive\\ integrity (\ref{def:inteti}/\ref{def:intecf})};
    \node[align=center] (te) at (3.5,-2) {transparent en-\\dorsement (\ref{def:te})};
    \node[align=center] (teit) at (0,-2) {termination-insen-\\sitive t.e. (\ref{def:teit})};

    \draw[double equal sign distance,-implies] (inte)--(inteti);
    \draw[double equal sign distance,-implies] (inte) to[bend left=45] (intewrong);
    \draw[double equal sign distance,-implies] (inte)--(te);
    \draw[double equal sign distance,-implies] (te)--(teit);
    \draw[double equal sign distance,-implies] (inteti)--(teit);

\end{tikzpicture}
\end{center}
We have not investigated the relationship between the wrong
definition of transparent endorsement (t.e.) proposed in Eqn. \ref{eqn:badte}
and the property in \ref{def:te}, but the remaining implications are
again exhaustive.

}

\section{Related Work}\label{rw}
The connection between modal logics and security properties has been studied before. Since Sutherland's work on non-deducibility \cite{Suth86}, 
a common trait to several works on information flow control has been the appeal to the concept of knowledge as a fundamental mechanism to bring out
what security property is being enforced and compare it with the knowledge allowed by the security policy \cite{dima2006nondeterministic,askarov2007gradual,DBLP:journals/corr/abs-1107-5594,DBLP:conf/nordsec/Balliu13,AhmadianB22,DBLP:conf/eurosp/McCallBJ22}. These works have produced elegant  security conditions for confidentiality and various flavours of declassification. Halpern and O'Neill \cite{DBLP:journals/tissec/HalpernO08} introduce a  framework for reasoning about confidentiality in multi-agent systems. Balliu et al. \cite{DBLP:conf/pldi/BalliuDG11} study epistemic modal logics to provide syntactical characterisations of confidentiality and declassification, and their relation to trace-based conditions  \cite{DBLP:conf/nordsec/Balliu13}.
Baumann et al. \cite{DBLP:conf/csfw/BaumannDGN21} introduce an epistemic approach to compare information flows of specifications and their refinements.
Clarkson et al. \cite{DBLP:conf/post/ClarksonFKMRS14} develop new logics for reasoning about trace-based conditions for hyperproperties \cite{DBLP:journals/jcs/ClarksonS10}. Recently, Lamport and Schneider  \cite{DBLP:conf/csfw/LamportS21} explore TLA to specify and verify a class of hyperproperties. Our work draws inspiration on these works, yet our contributions go beyond confidentiality properties. We provide a general and intuitive modal framework to characterise the interplay between confidentiality and integrity, declassification and endorsement, as well as subtle notions of termination- and progress-insensitivity which are at the heart of soundness justifications for modern enforcement mechanisms \cite{DBLP:journals/jsac/SabelfeldM03,DBLP:conf/esorics/AskarovHSS08,DBLP:journals/corr/abs-1107-5594}.

\changed{
One particularly noteworthy approach which \emph{does} unify confidentiality and integrity in a modal framework was proposed by Moore et al. \cite{mooreabstract}. In contrast with our proposal, that framework does not feature a modality of time at all, opting to represent entire runs as single possible worlds; its primitive modalities are a ``knows'' modality analogous to our $\K_A$, as well as a modality $[E_A]$ glossed as ``$A$ ensures'', which roughly corresponds to $[W_{\bar A}]$ in our framework, where $\bar A$ is the complementary agent of $A$, who can write everything that $A$ can not. Thus, $[E_A]\varphi$ if only $A$ can falsify $\varphi$ by performing a write.
To model security properties, which we think of as inextricably linked to the system's evolution over time (e.g. confidentiality saying that no agent gains additional knowledge), their work thus resorts to using the write capabilities as a proxy for ``inputs'' (i.e. state at time 0), so for example confidentiality is given by the formula $[E_A]\varphi \Rightarrow \neg [K_B]\varphi$ for all agents $A,B$. Formulae $\varphi$ that are ensured by $A$ are taken to be $A$'s inputs to the system, which a priori $B$ should not know.
The apparent dependency of a purely epistemic property talking about $B$'s knowledge on $A$'s active capabilities is unexpected, and has the consequence that $A$ active capabilities must be of a particular form.
This makes it impossible to express many policies in the seemingly obvious way. A particularly clear example of it arises for a policy that says every agent may know everything, but integrity of some data, say a variable $a$, must be maintained. Integrity is given dually to confidentiality as the formula $[K_A]\varphi \Rightarrow \neg [E_B]\varphi$ -- but if $[K_A]\varphi$ holds for every $\varphi$, this property says that $[E_B]\varphi$ may not hold for any $\varphi$ at all.
 To some extent, this issue can be sidestepped by restricting the sets of formulae and hand-crafting agents with $K$ and $E$ relations that do not directly correspond to the capabilities of any agent, but this makes the system unintuitive. Moreover, further difficulties arise when we consider properties with complex dependencies on time, such as (when-)declassification or explicit (where-)endorsement as discussed at the end of Section \ref{sec:inte}.
}

Other attempts at building such general frameworks include selective
interleaving functions \cite{McL94}, possibilistic security properties \cite{ZL97}, and Mantel's  assembly
kit \cite{Mantel00}. These approaches are quite different, and focus more on the modular construction of properties rather than
 extensional properties. Beyond information-flow properties, past works have studied modal logics, mainly epistemic logics \cite{rak}, in the context of 
 computer security, including BAN logic \cite{10.1145/77648.77649} and applied $\pi$-calculus \cite{DBLP:conf/forte/ChadhaDK09} to model
 knowledge in security protocols.

A key goal of our work is to illuminate on the different security properties proposed by the language-based security community, including robust declassification and transparent endorsement \cite{DBLP:conf/csfw/ZdancewicM01,DBLP:conf/csfw/MyersSZ04,DBLP:conf/csfw/ChongM06,DBLP:journals/jcs/MyersSZ06,DBLP:conf/pldi/BalliuM09,DBLP:journals/corr/abs-1107-5594,cecchetti2017nonmalleable,DBLP:conf/csfw/OakABS21}. A critique that may be levelled at the past work, our own included, is that it has not always managed to separate concerns very clearly. In particular, constraints in specification techniques, programming language features, and details and limitation in the enforcement mechanisms have been intertwined in such a way that it has often been unclear exactly what security properties are enforced and how these properties relate to each other. In contrast, we show that modal logic appears to be a well suited framework to study complex information flow properties in an elegant and intuitive manner. Cecchetti et al.  \cite{cecchetti2017nonmalleable} introduce the security property of nonmalleable information flow, which incorporates both robust declassification and  transparent endorsement in a trace-based setting. We use our modal framework to capture these properties in more general settings and show equivalence with appropriately generalised forms of both. 
The subtleties we uncover in the process demonstrate the advantage of our framework in exposing hidden assumptions in trace-based properties and allowing direct comparison between different definitions.
Askarov and Myers \cite{DBLP:journals/corr/abs-1107-5594} study knowledge-based security properties including robust declassification. Birgisson et al. \cite{DBLP:conf/iciss/BirgissonRS10} propose a trace-based framework for studying different facets of integrity. Our modal logic can serve as a unifying framework for these properties.

There exists a large array of works on verification of information flow properties \cite{DBLP:journals/jsac/SabelfeldM03}, including security type systems \cite{DBLP:journals/jcs/MyersSZ06,DBLP:conf/csfw/ChongM06,DBLP:journals/corr/abs-1107-5594} and program analysis \cite{DBLP:conf/pldi/BalliuM09,DBLP:conf/csfw/OakABS21} for robust declassification, and security type systems for nonmalleable security \cite{cecchetti2017nonmalleable}. Our work provides an intuitive baseline to justify soundness and security properties that are enforced by these mechanisms. A new direction for future work is the investigation of model checking techniques for direct verification of formulae in our modal logic. 
While model checking of modal logics in the context of security protocols has a long history \cite{GammieM04,lomuscio:2009:mcmas}, some works study model checking of programs with respect to specifications in epistemic logics  \cite{DBLP:conf/csfw/BalliuDG12,DBLP:conf/ijcai/GorogiannisRB17} and hyperproperty-based logics \cite{DBLP:conf/cav/FinkbeinerRS15}. Unfortunately, none of these works consider security properties that intertwine confidentiality and integrity.

\section{Conclusions}

We have introduced a framework based on modal logic which
allows us to reason about the security of computer systems by
representing the capabilities of agents in the system and
the permissions granted to them in the security policy as modal relations.
After showing how to represent programs and a standard class
of security policies %
in this framework, we have presented several intuitive
definitions of security properties from the literature, 
including variants of confidentiality, integrity, 
robust declassification and transparent endorsement, proving 
a variant of each equivalent to known trace-based definitions.
In the process, we have uncovered several subtleties of the established
definitions, including a potential issue with how termination-insensitive
robust declassification handles termination and an interesting
detail in how transparent endorsement deals with cause and effect.
The exposure of this sort of detail is one of the benefits of 
viewing existing properties through the lens of our framework.

We assumed program semantics to be deterministic, with any
non-deterministic behaviour optionally encoded as additional input parameters,
and all security properties we considered were possibilistic, 
rejecting programs if untoward information flows were possible regardless
of their likelihood. 
\changed{
In future work, we aim to consider
settings in which these assumptions do not hold,
such as differential privacy \cite{dwork-diffpriv},
as well as security properties based on quantitative information flow (QIF) \cite{denningsecurity,grayqif}.
We expect that
existing applications of modal logic to probability
and uncertainty (\cite{rau}, Chapter 7) will be of use to that end:
as a basic example, we may consider analogues of our definitions
where all knowledge modalities have been replaced by ones capturing
a level of certainty.
}

\changed{
While we have not focussed on formal verification, there is
extensive prior work and tooling available for model-checking security
properties with modal logics such as HyperLTL \cite{DBLP:conf/cav/FinkbeinerRS15}.
Security Kripke frames can be generated automatically for any system
with a known operational semantics, and there are no fundamental
obstacles to verifying our formulae directly on them. However,
the formulae that we presented are optimised for understandability
and power to facilitate equivalence proofs, rather than efficient
model checking. 
Investigating formal verification in our framework is a promising direction
for future work.
}

\section{Acknowledgements}

We are indebted to Mae Milano, who first suggested to the first author to investigate the application of modal logic to security properties, Owen Arden and Joe Halpern. Their early input and collaboration at Cornell was instrumental in the development of these ideas.
Thanks are also due to Andrew Myers, Fred Schneider, Scott Moore, Deepak Garg, Aslan Askarov, Andrei Sabelfeld  and David Sands for valuable discussion and feedback.
Finally, we would like to thank the anonymous reviewers for their insightful comments.

This work was partially supported by the Swedish Research
Council under projects JointForce and WebInspector, and the Swedish Foundation for Strategic
Research (SSF) under projects CHAINS and TrustFull.

\bibliography{biblio}{}

\iffullver

\appendix

\subsection{Missing proofs}

\newtheorem*{theorem*}{Theorem}
\newtheorem*{proposition*}{Proposition}

\begin{theorem*}[\ref{thm:confbasic}] Suppose $P$ is a program with s.s.c. $S$.
Then $P$ and $S$ satisfy Def. \ref{def:confts}
iff the Kripke interpretation of $P$ and $S$ satisfies
Def. \ref{def:conf}.
\end{theorem*}
\begin{proof}
\hypertarget{proof:confbasic}{}
\noindent \emph{Kripke $\Rightarrow$ trace-based.}
Want to show that a violation of definition \ref{def:confts}
implies a violation of Kripke frame confidentiality.
Suppose we have a pair of traces $t^i$ such that $\view_A(t^1)\neq \view_A(t^2)$,
but their initial states only differ in variables $A$ can not read.
Without loss of generality, assume that $t^1$ is at least as long
as $t^2$.
Let $w^i$ be the unique singleton traces that are prefixes
of the respective $t^i$. Then $w^1 \sim_{\K_A} w^2$ by definition.
Since $\view_A(t^1)\neq \view_A(t^2)$, the two lists
of observations must differ at some finite entry (or $\view_A(t^2)$,
being the shorter one, must end before some finite entry), say the $n$th.
Let $v^1$ be a successor of $w^1$ (i.e. $(w^1,v^1)\in T$) \todo{RG: Don't like the use of v}
at which $A$ has made at least $n$
observations, i.e. $|\view_A(v^1)|\geq n$.
Then for all $v^2$ s.t. $(w^2,v^2)\in T$, $v^1 \not\sim_{\K_A} v^2$.

Taking $\varphi$ to be the formula $\neg \varphi_{w^2}$, where
$\varphi_{w^2}$ is the characteristic formula for that run,
we therefore have $w^1 \vDash \Diamond \B{\K_A} \varphi$ (as
at $v^1$, we know that we are not in the run generated from $w^2$ \todo{ $v^1 \vDash \B{\K_A} \varphi$}),
but $w^1\vDash \neg \B{\K_A} \Diamond \varphi$, as $\varphi$
does not hold at any $T$-successor of $w^2$.

\vspace{0.5em}
\noindent \emph{trace-based $\Rightarrow$ Kripke.}
Want to show that a violation of definition \ref{def:conf}
implies a violation of trace-based confidentiality.
Suppose we are given a world $w$ such that
$w\vDash \Diamond \B{\K_A}\varphi \wedge \neg \B{\K_A}\Diamond \varphi.$
Let $w_1$ be the unique $T$-ancestor of $w$, i.e. world with $(w_1,w)\in T$,
such that $w_1$ is an initial trace, i.e. of length 1.
Then we must also have $w_1\vDash \Diamond \B{\K_A}\varphi \wedge \neg \B{\K_A}\Diamond \varphi$.
Indeed, $w_1\vDash \Diamond \B{\K_A}\varphi$ by transitivity of $T$.
To determine that $w_1 \vDash \neg \B{\K_A}\Diamond \varphi$,
we first observe that if $\varphi$ is t.s.,
then so is $\Diamond \varphi$. 
Hence if we had $w_1 \vDash \B{\K_A} \Diamond \varphi$,
then $w_1 \vDash \Box \B{\K_A}\Diamond \varphi$ by perfect recall,
and hence $w \vDash \B{\K_A}\Diamond\varphi$; but
$w \vDash \neg \B{\K_A}\Diamond \varphi$, resulting in a contradiction.

Since $w_1 \vDash \neg \B{\K_A}\Diamond \varphi$,
there must be a $w_2$ s.t. $w_1 \sim_{\K_A} w_2$ and $w_2 \vDash \Box \neg \varphi$.
All $T$-successors of $w_2$ have $\neg \varphi$.
On the other hand, since $w_1 \vDash \Diamond \B{\K_A} \varphi$,
there must be a $v_1$ s.t. $(w_1,v_1)\in T$ and
$v_1$ is not $\K_A$-related to any world where $\neg\varphi$.
In particular, $v_1$ may not be $\K_A$-related to any $T$-successor
of $w_2$. This implies that $A$'s view of the maximal trace
generated from $w_1$ has a prefix ($v_1$) which does
not occur as a prefix in $A$'s view of the maximal trace generated from $w_2$.
Therefore, the $\view_A$ of the two maximal traces must differ.

However, $w_1$ and $w_2$ themselves are $\K_A$-related,
so they only differ in $A$-unreadable variables, and as such satisfy
the preconditions of Definition \ref{def:confts}.
\end{proof}

\begin{theorem*}[\ref{thm:confeq}] Suppose $P$ is a program with s.s.c. $S$.
Then $P$ and $S$ satisfy Def. \ref{def:conf2}
iff the Kripke interpretation of $P$ and $S$ satisfies
Def. \ref{def:confti}.
\end{theorem*}
\begin{proof}
\hypertarget{proof:confeq}{}
Note that in this setting, $\KC=\KP$ for all agents, so we
will simply write $\K$.

\noindent \emph{Kripke $\Rightarrow$ trace-based.}
Want to show that a violation of definition \ref{def:conf2}
implies a violation of Kripke frame termination-insensitive confidentiality.
Suppose we have a pair of traces $t^i$ such that $\view_A(t^1)\neq \view_A(t^2)$,
and $t^1$ and $t^2$ both halt. Without
loss of generality, assume that $\view_A(t^1)$ is at least as long
as $\view_A(t^2)$.
\todo{The use of both t and w is confusing.}
Let $w^i$ be the unique singleton traces that are prefixes
of the respective $t^i$, and $w_\Downarrow^i$ \todo{why we don't use directly t?} be the shortest traces
where the respective $t^i$ has halted.
Then $\view_A(t^i)=\view_A(w_\Downarrow^i)$, as the observable
memory does not change after halting.
So $\view_A(w_\Downarrow^1) \neq \view_A(w_\Downarrow^2)$,
and in fact $\view_A(w_\Downarrow^1) \neq \view_A(w)$
for any $w$ s.t. $(w^2,w)\in T$, as those are either strictly shorter 
or equal to $w_\Downarrow^2$, and so $w_\Downarrow^1 \not\sim_{\K_A} w$.

At the same time, we have $w^1 \sim_{\K_A} w^2$ by definition \ref{def:conf2},
as their initial states only differ in variables that $A$ can not read.

Taking $\varphi$ to be the formula $\neg \varphi_{w^2}$, where
$\varphi_{w^2}$ is the characteristic formula for that run,
we therefore have $w^1 \vDash \Diamond \B{\K_A} \varphi$ (as
at $w^1_\Downarrow$, no $T$-successor of $w^2$ is $\K_A$-related),
and $w^1 \vDash \Diamond \Downarrow$ by assumption, 
but $w^1\vDash \neg \B{\K_A} (\Box \neg {\Downarrow} \vee \Diamond \varphi)$, as $\varphi$
does not hold at any $T$-successor of $w^2$, but $w^2\vDash \Diamond \Downarrow$.

\vspace{0.5em}

\noindent \emph{trace-based $\Rightarrow$ Kripke.}
Want to show that a violation of definition \ref{def:confti}
implies a violation of termination-insensitive trace confidentiality.
Suppose we are given a world $w$ such that
$$w\vDash \Diamond{\Downarrow} \wedge \Diamond \B{\K_A}\varphi \wedge \neg \B{\K_A}(\Box \neg {\Downarrow} \vee \Diamond \varphi).$$
Let $w_1$ be the unique $T$-ancestor of $w$, i.e. world with $(w_1,w)\in T$,
such that $w_1$ is an initial trace, i.e. of length 1.
Then we must also have 
$$w_1\vDash \Diamond{\Downarrow} \wedge \Diamond \B{\K_A}\varphi \wedge \neg \B{\K_A}(\Box \neg {\Downarrow} \vee \Diamond \varphi).$$
Indeed, $w_1\vDash \Diamond{\Downarrow} \wedge  \Diamond \B{\K_A}\varphi$ by transitivity of $T$.
To determine that $w_1 \vDash \neg \B{\K_A}\Diamond \varphi$, \todo{RG: Some of these things could become lemmata}
we first observe that if $\varphi$ is t.s.,
then so is $\Diamond \varphi$, and also $\Box \neg {\Downarrow} \vee \varphi$,
as $\Box\neg{\Downarrow}$ is true iff it is true everywhere in a run. 
Hence if we had $w_1 \vDash \B{\K_A} (\Box \neg {\Downarrow} \vee \Diamond \varphi)$,
then $w_1 \vDash \Box \B{\K_A} (\Box \neg {\Downarrow} \vee\Diamond \varphi)$ by perfect recall,
and hence $w \vDash \B{\K_A} (\Box \neg {\Downarrow} \vee\Diamond\varphi)$, 
contradicting the assumption about $w$.

Since $w_1 \vDash \neg \B{\K_A} (\Box \neg {\Downarrow} \vee\Diamond \varphi)$,
there must be a $w_2$ s.t. $w_1 \sim_{\K_A} w_2$ and $w_2 \vDash \Diamond{\Downarrow}\wedge \Box \neg \varphi$.
All $T$-successors of $w_2$ have $\neg \varphi$,
and in particular so does the first halted successor $w_{2,\Downarrow}$. \todo{RG There is only one halted successor}
On the other hand, since $w_1 \vDash \Diamond \B{\K_A} \varphi$,
by perfect recall, we have $w_{1,\Downarrow} \vDash \B{\K_A} \varphi$.
So $w_{1,\Downarrow}$ is not $\K_A$-related to any world where $\neg\varphi$,
in particular to $w_{2,\Downarrow}$.
But these two worlds are each complete terminating traces (extended from $w_1$ and $w_2$
respectively), and they differ in $\view_A$ by definition of $\K_A$. 
These two traces meet all preconditions necessary to violate
Def. \ref{def:conf2}.
\end{proof}

\begin{theorem*}[\ref{thm:inteeq}] Suppose $P$ is a program with s.s.c. $S$.
Then $P$ and $S$ satisfy Def. \ref{def:inte2}
iff the Kripke interpretation of $P$ and $S$ satisfies
Def. \ref{def:inteti}.
\end{theorem*}
\begin{proof}
\hypertarget{proof:inteeq}{}
\noindent \emph{Kripke $\Rightarrow$ trace-based.}
Want to show that a violation of definition \ref{def:inte2}
implies a violation of Kripke frame termination-insensitive integrity.
Suppose we have a pair of traces $t^i$ such that $\fix_A(t^1)\neq \fix_A(t^2)$,
and $t^1$ and $t^2$ both halt. Without
loss of generality, assume that $\fix_A(t^2)$ is at least as long
as $\fix_A(t^1)$.
Let $w^i$ be the unique singleton traces that are prefixes 
of the respective $t^i$, and $w_\Downarrow^i$ be the shortest traces
where the respective $t^i$ has halted.
Then $\fix_A(t^i)=\fix_A(w_\Downarrow^i)$, as the unwriteable
memory does not change after halting \todo{RG: There must be something here with infinite runs even if we halt that has not been spelled out, or I missed it}.
So $\fix_A(w_\Downarrow^1) \neq \fix_A(w_\Downarrow^2)$,
and in fact $\fix_A(w) \neq \fix_A(w_\Downarrow^2)$
for any $w$ s.t. $(w^1,w)\in T$, as those are either strictly shorter

At the same time, we have $w^1 \sim_{\WC_A} w^2$ by definition \ref{def:inte2},
as their initial states only differ in variables that $A$ can not read. \todo{RG: can write}

Let $\varphi_{w^2}$ be the characteristic formula for the run of $w^2$,
$\tau$ be the length of $w^2_\Downarrow$, $a$ an arbitrary variable
and $v$ the value it takes in $w^2_\Downarrow$, so $a@\tau=v$ 
is true at $w^2_\Downarrow$ and false before it.
\todo{RG: ?}
be any formula that only 
Taking $\varphi$ to be the formula $\varphi_{w^2} \wedge a@\tau=v$, 
we therefore have $w^1 \vDash \langle \WC_A\rangle (\Diamond{\Downarrow} \wedge \Diamond \varphi)$
(as $w^2$ is reachable by $\WC_A$, that run halts, and $\varphi$ is true when it does),
but $w^1 \vDash \neg\Box\neg{\Downarrow} \wedge \neg \Diamond\langle \WP_A\rangle \varphi$
 (as the run halts, but $w^2_\Downarrow$ is never $\WP_A$-reachable from it).

\vspace{0.5em}

\noindent \emph{trace-based $\Rightarrow$ Kripke.}
Want to show that a violation of definition \ref{def:inteti}
implies a violation of trace-based integrity.
Suppose we are given a world $w_1$ such that
$$w_1\vDash \langle \WC_A\rangle (\Diamond{\Downarrow} \wedge \Diamond \varphi) \wedge \Diamond{\Downarrow}
\wedge \neg\Diamond\langle \WP_A\rangle \varphi. $$
Then $w_1$ must be an initial trace, i.e. of length 1,
as otherwise $\WC_A$ only relates it to itself, but we can not have
$\Diamond\varphi$ and $\neg \Diamond \langle \WP_A\rangle \varphi$ at the same world
since $\WP_A$ itself is reflexive.

Let $w_2$ be the $\WC_A$-related world at which $w_2 \vDash \Diamond{\Downarrow} \wedge \Diamond \varphi$.
Since $\varphi$ is t.s., we must in fact
have $\Diamond ({\Downarrow}\wedge \varphi)$, i.e. $\varphi$
is satisfied at the successor of $w_2$ where that run first halts, $w_{2,\Downarrow}$.
Also, the run generated from $w_1$ halts, say first at $w_{1,\Downarrow}$, and because we have
$w_1\vDash \Box \neg \langle \WP_A\rangle \varphi$,
we also have $w_{1,\Downarrow}\vDash \neg \langle \WP_A\rangle \varphi$.
In particular, $\fix_A (w_{1,\Downarrow}) \neq \fix_A (w_{2,\Downarrow})$.
But these two worlds are each complete terminating traces (extended from $w_1$ and $w_2$
respectively, which are $\WC_A$-related).
These two traces meet all preconditions necessary to violate
Def. \ref{def:inte2}.

\end{proof}

\begin{theorem*}[\ref{thm:syncrd}]
Suppose $P=\langle p, V, I\rangle $ signals termination and $S$ is an s.s.c. defined by $\langle W, R\rangle$, with $W(A)\subseteq R(A)$ for all agents $A$.
Then $P$ satisfies trace-based robust declassification (Def. \ref{def:rd2})
iff the Kripke interpretation of $P$ satisfies termination-insensitive robust declassification (Def. \ref{def:rdit}).
\end{theorem*}
\begin{proof}
\hypertarget{proof:syncrd}{}
\noindent \emph{Kripke $\Rightarrow$ trace-based}.
Want to show that a violation of definition \ref{def:rd2} implies a violation of Kripke frame robust declassification.
Suppose we have a quadruple of traces $t^{ii}$ such that WLOG $\view_A(t^{11})=\view_A(t^{21})$, but $\view_A(t^{12})\neq \view_A(t^{22})$, and all terminate.

Let $w^{ii}$ be the unique singleton trace that's a $T$-ancestor of $t^{ii}$.
Then we have $w^{12} \sim_{\K_A} w^{22}$, $w^{11} \sim_{\K_A} w^{21}$, $w^{11} \sim_{\WC_A} w^{12}$ by definition.
Also, since all four maximal traces terminate, we in fact have
$\view_A(w^{11}_\Downarrow)=\view_A(w^{21}_\Downarrow)$ and $\view_A(w^{12}_\Downarrow)\neq \view_A(w^{22}_\Downarrow)$.

Let $\varphi$ be the formula that encodes that we are not in a run 
reachable from $w^{21}$ by $\WC_A$, i.e. $\neg (\varphi_{w^{21}} \vee \varphi_{w^{22}} \vee \ldots)$.
Then $w_\Downarrow^{12}\vDash \B{\K_A} \varphi$, as $w_\Downarrow^{12}$ is not $\K_A$-related to any successor of $w^{21}$'s $\WC_A$ relatives.
Hence, $w^{12} \vDash \Diamond \B{\K_A}\varphi$. However, $w^{12}\vDash \neg \B{\K_A}(\Box\neg{\Downarrow}\vee\Diamond \varphi)$, as $w^{12}$ is $\WC_A$-related to $w^{22}$, whose run also halts.
Also, $w^{12}\vDash \B{\K_A}(\Box \neg\varphi \Rightarrow \Diamond{\Downarrow})$, as the only $\K_A$-related 
world where $\Box\neg\varphi$ holds is $w^{22}$, whose run also halts.
Thus, $w^{11} \vDash \langle \WC_A \rangle (\Diamond{\Downarrow}\wedge \Diamond \B{\K_A}\varphi \wedge\neg \B{\K_A}(\Box\neg{\Downarrow}\vee\Diamond \varphi) \wedge \B{\K_A} (\Box\neg\varphi \Rightarrow \Diamond{\Downarrow})) $.

However, $w^{11} \vDash \Diamond{\Downarrow}\wedge \neg \Diamond \B{\K_A} \varphi$, as if $\B{\K_A}\varphi$ at any $T$-successor, then also
$w^{11}_\Downarrow\vDash \B{\K_A}\varphi$ by perfect recall. But $w^{11}_\Downarrow$ is $\K_A$-related to $w^{21}_\Downarrow$, at which $\varphi$ is false. 
Thus Def. \ref{def:rd1} is violated.

\vspace{0.5em}

\noindent \emph{trace-based $\Rightarrow$ Kripke}.
Want to show that a violation of Definition \ref{def:rd1} implies a violation
of trace-based robust declassification.
Suppose we have
\begin{align*} w\vDash
& \langle \WC_A \rangle ( \Diamond{\Downarrow} \wedge \Diamond \B{\K_A}\varphi \wedge \neg \B{\K_A}(\Box\neg{\Downarrow} \vee \Diamond \varphi ) \\ 
 & \hspace{2em} \wedge \B{\K_A} (\Box\neg\varphi \Rightarrow \Diamond{\Downarrow}) ) \\
& \Rightarrow \Box\neg{\Downarrow} \vee \Diamond \B{\K_A}\varphi,
\end{align*}
or equivalently 
\begin{align} w\vDash
& \langle \WC_A \rangle ( \Diamond{\Downarrow} \wedge \Diamond \B{\K_A}\varphi \wedge \neg \B{\K_A}(\Box\neg{\Downarrow} \vee \Diamond \varphi ) \nonumber \\ 
 & \hspace{2em} \wedge \B{\K_A} (\Box\neg\varphi \Rightarrow \Diamond{\Downarrow}) ) \nonumber \\
& \wedge \Diamond{\Downarrow} \wedge \neg\Diamond \B{\K_A}\varphi. \label{eqn:worldw}
\end{align}
We have $w\vDash \Box\neg\B{\K_A}\varphi$, and $w\vDash \Diamond{\Downarrow}$. So
at the terminating $T$-successor $w_\Downarrow$, we still have $\neg\B{\K_A}\varphi$,
and it is related to a world where $\neg\varphi$ and (as $P$ signals termination) $\Downarrow$,
say $w_{2,\Downarrow}\vDash {\Downarrow} \wedge\neg\varphi$. This world must have a unique singleton
$T$-ancestor $w_2$, where $w_2\vDash \Diamond{\Downarrow}\wedge \Box\neg\varphi$.

Let $w_1$ denote the $\WC_A$-reachable world at which
$w_1 \vDash \Diamond{\Downarrow} \wedge \Diamond \B{\K_A}\varphi \wedge \neg \B{\K_A}(\Box\neg{\Downarrow} \vee \Diamond \varphi ) \wedge \B{\K_A} (\Box\neg\varphi \Rightarrow \Diamond{\Downarrow})$.
By commutativity, there must be a $w_3$ with $w_1 \sim_{\K_A} w_3 \sim_{\WC_A} w_2$,
where by write-stability between $w_2$ and $w_3$ we also have $w_3\vDash \neg\Diamond\varphi$.
Since $w_1\vDash \B{\K_A} (\Box\neg\varphi \Rightarrow \Diamond{\Downarrow})$,
we must have $w_3 \vDash \Diamond{\Downarrow}$.
Also, $w_1\vDash \Diamond \B{\K_A}\varphi$, so at least at $w_{1,\Downarrow}$,
we must have $\B{\K_A}\varphi$ and hence $w_{1,\Downarrow}\neg\sim_{\K_A} w_{3,\Downarrow}$,
since $w_{3,\Downarrow}\vDash \neg\varphi$.
So the runs generated by $w_1$ and $w_3$ are distinguishable in $A$'s view, but the runs generated by $w$ and $w_2$ are not, and all four terminate. Done.
\end{proof}

\begin{proposition*}[\ref{prop:teeq}] \begin{enumerate}[(i)]
\item If $P$ is a program with s.s.c. $S$, then the Kripke interpretation of $P$ and $S$ satisfies
cause termination-insensitive integrity iff $P$ and $S$ satisfy trace-based integrity (Def. \ref{def:inte2}).
\item If $P$ also signals termination and $W(A)\subseteq R(A)$, then it satisfies trace-based transparent endorsement (Def. \ref{def:te2})
iff the Kripke interpretation of $P$ and $S$ satisfies termination-insensitive transparent endorsement (\ref{def:teit}).
\end{enumerate}
\end{proposition*}
\begin{proof} \hypertarget{proof:teeq}{} (Sketch)
The proof can proceed analogously to Theorems \ref{thm:confeq} and \ref{thm:syncrd}, despite the permission-capability swap, as we only ever use $\KC_A$ productively on initial worlds, which are also those on which $\WP_A$ and $\WC_A$ agree. In the proof of (ii), some additional care is necessary to address $\WP_A\neq \WC_A$. Note that the condition $W(A)\subseteq R(A)$ is, counterintuitively, self-dual, as $W(A)$ refers to the \emph{complement} of the variables whose values are unchanged by $\fix_A$.
\end{proof}

\fi %

%

\end{document}